\documentclass[11pt]{article}
\usepackage[letterpaper,margin=1in]{geometry}
\usepackage{amsmath,amssymb,amsthm}
\usepackage{adjustbox}
\usepackage{microtype}
\usepackage{xcolor}
\usepackage[colorlinks=true, allcolors=blue]{hyperref}
\usepackage{booktabs}
\usepackage{lmodern}
\usepackage[ruled]{algorithm2e}

\SetAlFnt{\small}
\SetAlCapFnt{\small}
\SetAlCapNameFnt{\small}
\SetAlCapHSkip{0pt}
\IncMargin{-\parindent}

\usepackage{enumitem}
\usepackage{makecell}

\newcommand{\PP}{\mathbb{P}}

\newcommand{\calF}{\mathcal{F}}

\newcommand{\RR}{\mathbb{R}}
\newcommand{\NN}{\mathbb{N}}

\newcommand{\wins}{\mathrm{Wins}}

\DeclareMathOperator{\AUX}{\mathsf{AUX}}
\DeclareMathOperator{\OPT}{\mathsf{OPT}}

\DeclareMathOperator{\ALG}{\mathsf{ALG}}

\DeclareMathOperator{\Bin}{\mathsf{Bin}}
\DeclareMathOperator{\Unif}{\mathsf{Unif}}

\DeclareMathOperator{\GREEDY}{\mathsf{GREEDY}}
\DeclareMathOperator{\spa}{\mathrm{span}}

\title{Secretary Problems with Interactive Ordinal Queries}

\author{José A. Soto \thanks{Center for Mathematical Modeling CNRS-IRL 2807, Universidad de Chile.} \thanks{Department of Mathematical Engineering, Universidad de Chile.} \and Felipe Valdevenito~\footnotemark[2]}
\date{}

\newtheorem{theorem}{Theorem}[section]
\newtheorem{lemma}[theorem]{Lemma}
\newtheorem{claim}[theorem]{Claim}
\newtheorem{definition}[theorem]{Definition}

\newtheorem{proposition}[theorem]{Proposition}
\newtheorem{remark}[theorem]{Remark}
\begin{document}

\maketitle
\begin{abstract}
We study secretary problems in which the online algorithm can ask a reliable oracle a limited number of ordinal queries about unseen elements. Values are revealed only when elements arrive, and all accept/reject decisions are irrevocable. The oracle never reveals numerical values; it only answers questions defined relative to the information already observed by the algorithm.

We first consider the single-choice secretary problem. A query returns the number of unseen elements whose value exceeds the current record. We analyze random order, adversarial order, and two sample-based adversarial orders, with no queries, one query, and unlimited queries. We obtain tight bounds for all variants except one. In adversarial order, queries do not improve over the trivial $1/n$ guarantee. In random order, one query improves the optimal success probability from $1/e$ to $0.4477\ldots$, and unlimited queries achieve the optimal value $1/2$. In the sample-based models, one query already yields the tight value $1/e$ against an online adversary, while unlimited queries give the exact value $(20\sqrt{10}-29)/81$ against an offline adversary. The remaining one-query offline-adversarial case is left open, and we give two partial upper bounds for it.

We also extend the model to the matroid secretary problem, where we study two ordinal query interfaces. A simple query identifies the unseen elements that would individually improve the current optimum on the observed set, while a complete query gives a finer ordinal partition of the unseen elements by weight relative to that optimum. Under sample-based online adversarial order, a single simple query yields a $1/4$ probability-competitive algorithm, and a complete query improves this to the optimal $1/e$. Finally, with unlimited complete queries in the random-order model, we achieve the optimal $1/2$ probability-competitive ratio.

\end{abstract}
\newpage
\vspace{1cm}
\setcounter{tocdepth}{2}

\section{Introduction}

In the Single-Choice Secretary Problem, a collection of $n$ candidates arrives in a uniformly random order. At each step, the decision-maker (DM) observes the candidate's rank among those seen so far and must irrevocably decide whether to accept the candidate and stop, or to reject the candidate and continue. Equivalently, each candidate has a private positive weight revealed upon arrival, and the goal is to select the maximum-weight candidate. It is well known (see Lindley~\cite{Lindley1961} and Dynkin~\cite{Dynkin1963}) that for every $n$ there exists a strategy with success probability at least $1/e$, and that no algorithm can guarantee a success probability exceeding $1/e$ for all $n$.

The Matroid Secretary Problem~\cite{BabaioffIKK2018} generalizes the Single-Choice Secretary Problem by imposing feasibility constraints. The candidates form the ground set of a matroid. Each element has a private positive weight revealed upon arrival, and the decision-maker seeks to select an independent set of maximum total weight. As in the classical problem, elements arrive in a uniformly random order. Upon observing an element, the decision-maker must irrevocably decide whether to accept it while maintaining independence in the matroid.

The performance of an algorithm for the Matroid Secretary Problem is typically measured by its \emph{utility-competitive ratio}, defined as the expected weight of the output independent set divided by the optimum. A stronger \emph{probability-competitive} notion is also used in the literature. Fix a maximum-weight independent set $\OPT$ by breaking ties arbitrarily. An algorithm is \emph{$\alpha$-probability-competitive} if $\PP[e \text{ is selected}] \ge \alpha$ for all $e\in\OPT$. This implies a utility-competitive ratio of at least $\alpha$.

The matroid secretary conjecture of Babaioff, Immorlica, and Kleinberg remains open. It asserts that every matroid admits a constant-factor utility-competitive algorithm in the random-order model. For general matroids of rank $r$, the best known utility-competitive guarantee is $\Omega(1/\log\log r)$ (see, e.g., \cite{FeldmanSZ18}). For the stronger probability-competitive notion, the best known guarantee is $\Omega(1/\log r)$ (see, e.g., \cite{SotoTV2021}).

A feature of both the Single-Choice and Matroid Secretary Problems is that the decision-maker has no information about the candidates' weights before they arrive. Natural variants endow the decision-maker with additional information. For example, in the Prophet Secretary problem each candidate's weight is drawn from a distribution over the nonnegative reals, and all these distributions are known to the decision-maker in advance. Realized weights are still revealed only upon arrival in a uniformly random order.

With this distributional information, one can obtain stronger guarantees. In the single-choice case, the best known competitive ratio improves to $0.688$ in the Prophet Secretary setting~\cite{Chen0LT25}, improving upon previous bounds such as $0.669$~\cite{CorreaSZ2021}. In the Matroid Prophet Secretary Problem, one can achieve a utility-competitive ratio of at least $1-1/e$~\cite{Ehsani2018}, which is substantially larger than the known guarantees without distributional knowledge.

Additional information can also be provided in other ways. For instance, one may use samples~\cite{azar2014prophet,CorreaCFOT21,CorreaCES22} or data-driven historical data~\cite{KaplanNR20} (see also~\cite{DuttingLLV21}). Another approach uses machine-learned predictions of numerical values, such as the maximum weight~\cite{AntoniadisGKK23,FujiiYoshida2024,BraunS24,KarisaniEtAl2026}. The goal is to design algorithms whose performance improves with prediction accuracy.

A different line of work studies secretary problems under ordinal information rather than numerical values. Hoefer and Kodric~\cite{HoeferK17} study combinatorial secretary problems in which the algorithm observes only ordinal rankings. Soto, Turkieltaub, and Verdugo~\cite{SotoTV2021} study probability-competitive guarantees for ordinal matroid secretary algorithms. More recently, Nourmohammadi, Cao, Sun, and Tan~\cite{NourmohammadiCST26} study ordinal secretaries with advice about the position of the best candidate. These models restrict or augment the information revealed by arriving elements. Our model is different: the algorithm receives the usual ordinal information from observed elements, and it may also ask interactive ordinal questions about the unseen set after the current decision has already been made.

In this paper, we study an interactive model inspired by explorable uncertainty~\cite{Erlebach2015}. Instead of receiving numerical predictions or upfront advice, the decision-maker adaptively acquires information during the online process through a trusted oracle. Each oracle call returns only ordinal information defined from the state that is already public at query time. The oracle does not disclose exact future weights, nor does it identify which unseen element is globally optimal. The oracle evaluates ordinal relations between the observed elements and the unseen set, such as whether an unseen element exceeds the current record. We assume the oracle is reliable to isolate the value of interactive ordinal information.

This query model naturally captures screening processes with sensitive evaluations. Consider a hiring process where an external auditor holds confidential candidate scores. The decision-maker observes a score only upon interviewing a candidate. Between interviews, the decision-maker can ask the auditor a limited number of ordinal queries. In the single-choice problem, one query returns the number of unseen candidates whose scores exceed the current record. This reveals how many future candidates lie above the current threshold, but it does not reveal their scores.

This screening framework extends to feasibility constraints, such as hiring subject to per-category quotas (a partition matroid). After interviewing a subset of candidates, the decision-maker identifies the optimal feasible team among them. A simple query round then asks the auditor to identify which unseen candidates would individually improve this reference team. The auditor certifies these improving candidates using only ordinal comparisons against the current optimal team, without revealing their absolute scores.

\paragraph{Query model.}
Current-candidate queries for secretary problems were studied by Liu, Milenkovic, and Moustakides~\cite{LiuMM23Mallows} under the Mallows model. In their model, the algorithm has a limited number of opportunities to ask an infallible expert whether the candidate currently under consideration is the global best. This type of query is closely related to the classical $k$-chance problem of Gilbert and Mosteller~\cite{GilbertM66}, where one has $k$ opportunities to stop on the best element. In the single-choice setting, a current-candidate query opportunity behaves like one such chance: a negative answer spends one opportunity and lets the search continue, while a positive answer certifies success. Moustakides, Liu, and Milenkovic~\cite{MoustakidesLM24RandomQueries} study a related optimal-stopping model with random, possibly faulty, expert responses.

Our queries have a different timing and a different purpose. They can be asked only after the current element has been irrevocably rejected and before the next element arrives. Thus, they provide information about unseen elements whose weights exceed a threshold determined by the public history. However, since they are asked between arrivals, they cannot certify the current stopping decision. In particular, even when a response $\mathrm{QR}_i=0$ certifies that the rejected element was globally best, this certificate arrives after the decision has been made. This distinguishes our model from current-candidate queries and from the $k$-chance/Dowry problem.

\subsection{Our results}

We begin with the single-choice secretary problem, which cleanly isolates the value of interactive ordinal queries. We consider four arrival models:
(i) \emph{Random Order} (RO), where the arrival order is a uniformly random permutation;
(ii) \emph{Adversarial Order} (AO), where an oblivious adversary fixes the order before the run;
(iii) \emph{Sample-Based Online Adversarial Order} (SBOnAO), where the algorithm first sees a uniformly random sample and then an \emph{online} adversary orders the remaining elements after seeing that sample; and
(iv) \emph{Sample-Based Offline Adversarial Order} (SBOffAO), where the adversary fixes a preliminary order \emph{before} the random sample is realized, and the post-sample arrivals follow the induced order.

Across these arrival models, we study three query budgets (0, 1, and unlimited query rounds), yielding $12$ single-choice variants. Our main qualitative conclusions are:
(i) in AO, even against an oblivious adversary, ordinal queries do not help and the optimal guarantee stays $1/n$;
(ii) in RO, one query round improves the classical $1/e$ success probability to the optimal value $\approx 0.4477$, and unlimited rounds achieve the optimal $1/2$; and
(iii) in the sample-based models, a single query round already exposes which unseen elements lie above the post-sample threshold, yielding tight bounds for SBOnAO. For SBOffAO with one query round, the unrestricted optimal value remains open. We prove an optimal $1/e$ bound under the above-threshold arrivals restriction of Section~\ref{sec:SBOffAO}, and we give a separate upper bound for immediate-query algorithms.

We then extend these ideas to matroid constraints. We consider two forms of ordinal query access, \emph{simple queries} (SQ) and \emph{complete queries} (CQ). Unlike the single-choice count query, these matroid query responses may identify subsets of unseen elements and may have size $\Theta(n)$. They remain ordinal because the oracle only certifies membership in sets defined by comparisons with the current optimal solution on the observed elements.

Throughout the paper, we measure complexity in query rounds rather than communication bits. This distinction is crucial because a single matroid query response can have size $\Theta(n)$. Thus, our bounds apply strictly to a query-round model and do not imply low communication complexity. Under this model, the ordinal information from the queries is sufficient to achieve constant probability-competitive guarantees in several settings, notably SBOnAO with one query round and RO with unlimited complete queries.

The tables below summarize our bounds. Entries marked with an asterisk ($*$) are derived from known no-query bounds and monotonicity between models (harder downward; easier to the right). All other entries are proved directly in this paper.

\begin{table}[h!]
\begin{adjustbox}{max width=\linewidth}
\begin{tabular}{|l|l|l|l|}
\hline
         & No Queries & 1 Query Round     & Unlimited Query Rounds \\ \hline
RO       & $=1/e$     & $=0.4477\ldots$ (Theorem~\ref{theo:ro1})     & $=1/2$ (Theorem~\ref{theo:rounl})             \\ \hline
SBOffAO & $=1/4$ (Theorem~\ref{theo:off})     & open; $\ge 1/e$ & $=(20\sqrt{10}-29)/81$ (Theorem~\ref{theo:offunl})\\ \hline
SBOnAO  & $=1/4$     & $=1/e$ (Theorem~\ref{theo:on1})              & $=1/e$ (Theorem~\ref{theo:on1})             \\ \hline
AO       & $=1/n$     & $=1/n$ (Theorem~\ref{theo:ao1})              & $=1/n$ (Theorem~\ref{theo:ao1})             \\ \hline
\end{tabular}
\end{adjustbox}
\caption{Lower and upper bounds on the competitive ratios for Single-Choice Secretary Problem models. For SBOffAO$1$, the unrestricted optimal value remains open. Under the \emph{above-threshold arrivals restriction} defined in Section~\ref{sec:SBOffAO}, the optimal value is $1/e$ (Theorem~\ref{theo:off1indist}); for immediate-query algorithms, Theorem~\ref{theo:off1} gives the upper bound $0.3787$.}
\end{table}

Caveat: in the matroid table, the no-query entries are \emph{utility}-competitive (from the literature), while our query-model entries are \emph{probability}-competitive.

\begin{table}[h!]
\begin{adjustbox}{max width=\linewidth}
\begin{tabular}{|l|l|l|l|l|}
\hline
             & No Queries       & 1SQ & 1CQ & $\infty$CQ \\ \hline
RO       & $\geq\Omega(1/\log\log r)$ & $\geq1/4*$      & $\geq1/e*$        & $=1/2$ (Theorem~\ref{theo:mspcunl})                   \\ \hline
SBOffAO & $\geq\Omega(1/\log\log r)$ & $\geq1/4*$      & $\geq1/e*$        &                   $\geq 1/e*$        \\ \hline
SBOnAO  & $\geq\Omega(1/\log\log r)$ & $\geq1/4$ (Theorem~\ref{theo:msps1})       & $=1/e$ (Theorem~\ref{theo:mspc1})           &          $=1/e*$                 \\ \hline
AO           & $=1/n$           & $=1/n*$         & $=1/n*$           & $=1/n*$                   \\ \hline
\end{tabular}
\end{adjustbox}
\caption{Bounds for Matroid Secretary Problem models. All query-model entries proved in this paper are probability-competitive guarantees. The no-query $\Omega(1/\log\log r)$ entries are utility-competitive baseline guarantees from the literature and are included only for comparison; they are not probability-competitive bounds.}
\label{table:msp}
\end{table}

We now turn to the formal model and analysis, beginning with the single-choice secretary problem.

\section{Single-choice secretary problems with ordinal queries}

In the Classical Secretary Problem, a known number of candidates arrives in a uniformly random order to be interviewed by an employer who seeks to hire the best candidate. Immediately after an interview, the employer must either hire the candidate or reject them, with no possibility of calling them back later.

\paragraph{Ordinal algorithm assumptions.}
Our goal is to understand how well we can approximate (i) selecting the maximum weight from a set and (ii) selecting a maximum-weight independent set in a matroid in online models with ordinal query access, compared to the corresponding offline optima. Both problems admit efficient \emph{ordinal} offline algorithms, i.e., algorithms that do not perform arithmetic on numerical weights but use only comparisons. Accordingly, and because our oracle provides only ordinal query responses, we restrict attention throughout to ordinal online algorithms.

Although ordinal algorithms do not require numerical weights, it is often convenient to describe them in terms of weights, as long as it is clear that they can be implemented using only comparisons. Any ordinal algorithm must behave identically on any two weight functions that induce the same order over the elements. In the single-choice models, element labels carry no information before weights are revealed. Thus the algorithm may use the arrival position of an element, the ordinal comparisons among revealed weights, and the oracle responses, but it may not use an element name as a proxy for its hidden weight. Equivalently, before weights are revealed, elements are indistinguishable except through the information exposed by the online process.

In this paper, we study variations of the Secretary Problem obtained by combining two independent modifications.

\begin{enumerate}
    \item Changing the order in which candidates arrive.
    \item Giving the decision-maker access to an oracle that provides ordinal queries about candidates who have not yet been interviewed.
\end{enumerate}

We now formalize these choices, each defining a different \emph{model}. Let $E$ be a finite set of candidates with $n := |E|$, and let an injective weight function $w: E \rightarrow \mathbb{R}_{>0}$ be chosen adversarially before the arrival order is determined. The algorithm receives the elements of $E$ one by one in an order $\pi: [n] \rightarrow E$ that depends on the model, where $\pi(j)$ denotes the $j$-th arriving candidate.

Each time an element $x \in E$ arrives, its weight $w(x)$ is revealed, at which point the algorithm must decide whether to select $x$ and stop, or to reject $x$ and proceed to the next element. Throughout the paper, we assume that $E=\{x^1,x^2,\dots, x^n\}$ is indexed in decreasing order of weight, $w(x^1)>w(x^2)>\dots > w(x^n)$.

We say that an algorithm $A$ for a model $\mathcal{M}$ \emph{wins} if it selects the heaviest element. The competitiveness of $A$ in model $\mathcal{M}$ is
\[
\alpha_\mathcal{M}(A) := \inf_{I}\; \PP[A \text{ wins on instance } I],
\]
where the probability is over the randomness of both the algorithm and the model, and the infimum ranges over all instances $I$ of the model with fixed input size $n$ (we suppress the dependence on $n$ in the notation). The competitiveness of the model $\mathcal{M}$ is $\alpha_\mathcal{M} := \sup_A \alpha_\mathcal{M}(A)$.

When we state a constant competitive ratio $\alpha$ (i.e., independent of $n$), we mean that for every $n$ and every instance of size $n$, the algorithm succeeds with probability at least $\alpha$.

To prove a model-level upper bound $C$ (as a constant), it suffices to show that for every algorithm $A$ and every $\varepsilon>0$, there exist arbitrarily large $n$ and an instance of size $n$ on which $A$ succeeds with probability at most $C+\varepsilon$; otherwise, an algorithm achieving $C+\varepsilon$ on all sufficiently large $n$ would contradict this bound.

\subsection{Models: Arrival Order}

We consider four types of arrival order: random (RO), adversarial (AO), sample-based online adversarial (SBOnAO), and sample-based offline adversarial (SBOffAO).

In Random Order models (RO models), the arrival order is drawn uniformly at random from all permutations of $E$. In particular, for any algorithm $A$, we define
$\alpha_{\mathrm{RO}}(A) := \PP_{\pi\sim\mathrm{Unif}}[A\; \wins]$,
where the probability is over the random arrival order $\pi$ and any internal randomness of $A$.

In Adversarial Order models (AO models), the order is chosen by an oblivious adversary who may know the algorithm but fixes the arrival order before the algorithm's random choices are realized. For every algorithm $A$, we define
$\alpha_{\mathrm{AO}}(A) := \min_{\sigma}\PP[A\; \wins \mid \pi=\sigma]$,
where the minimum is taken over all orders $\sigma$ of $E$.
Thus the adversary minimizes over fixed permutations and does not adapt to the algorithm's observations or oracle queries during the run.

The Sample-Based Online Adversarial Order models (SBOnAO models) and the Sample-Based Offline Adversarial Order models (SBOffAO models) allow the algorithm to observe a sample of candidates before the online process begins. More precisely, an algorithm $A$ first declares a sample size $s\in \{0,1,\dots, n\}$ (possibly at random). The algorithm then enters a \emph{sampling phase}, in which it receives and observes a uniformly random subset $S \subseteq E$ of size $s$, from which it cannot select any element. Finally, it enters a \emph{selection phase}, in which the elements in $E\setminus S$ arrive one by one according to the rules of the model.
In all adversarial models, the adversary may know the algorithm and the instance. In SBOnAO, the adversary may also know the realized sample before fixing the post-sample order. However, the adversary does not observe the algorithm's private random bits before making its ordering decisions.

Algorithms for both SBOnAO and SBOffAO can be viewed as tuples $(\mathcal{D},A')$, where $\mathcal{D}$ is a distribution over sample sizes with support in $\{0,\dots,n\}$ (more precisely, for each $n$ there is a distribution $\mathcal{D}_n$ from which $s$ is drawn), and $A'$ is the \emph{post-sample algorithm}. We interpret this tuple as follows: the sample size $s$ is drawn from $\mathcal{D}$, and $A'$ is the part of the algorithm that processes the elements arriving after the sample.

The difference between SBOnAO and SBOffAO is how the elements outside the sample are ordered in the selection phase. In SBOnAO, the elements of $E\setminus S$ arrive one by one in an adversarial order that may depend on $S$; this order is chosen by an \emph{online adversary} that observes the realization of $S$. In SBOffAO, the order is chosen by an \emph{offline adversary}. That is, before the experiment starts, the offline adversary fixes a \emph{preliminary order} $\psi\colon [n] \rightarrow E$ that is independent of the realized sample (it may depend on $\mathcal{D}$, but not on the realized sample size $s$). This order is unknown to the algorithm. In the selection phase, the elements of $E \setminus S$ arrive in the unique order $\phi:[n-s]\rightarrow E\setminus S$ consistent with $\psi$; that is, for all $x,y\in E\setminus S$, $\phi^{-1}(x)<\phi^{-1}(y)$ if and only if $\psi^{-1}(x)<\psi^{-1}(y)$.

\subsection{Models: Query rounds}

In addition to the arrival-order restrictions, we consider an algorithm that can make \emph{ordinal} queries during the run. The query is defined relative to the current record (the best weight observed so far) and is only allowed \emph{between} arrivals: after rejecting the element at step $i$ and before seeing the element at step $i+1$, the algorithm may query an oracle. We call such an oracle call a \emph{query round}.

The oracle returns a count rather than a single bit. Under our timing restriction, the bit ``is there an unseen element better than the current record?'' is not useful: if the answer is no, then the (already rejected) record was globally best; if the answer is yes, the algorithm only learns that some better element remains. The count $\mathrm{QR}_i$ captures the relevant ordinal information---how many above-record elements remain---without revealing their identities or numerical weights.

For an observed set $F\subseteq E$, define
\[
\tau(F):=
\begin{cases}
\max_{y\in F}w(y) & \text{if } F\ne\emptyset,\\
0 & \text{if } F=\emptyset.
\end{cases}
\]

\begin{definition}[Single-choice query response]
For a given arrival order $\pi$, weight function $w$, and $i\in\{0,1,\dots,n\}$, the query response at step $i$, denoted $\mathrm{QR}_i$, is defined as
\[
\mathrm{QR}_i := |\{x\in E\setminus E_i: w(x) > \tau(E_i)\}|,
\]
where $E_i$ is the set of the first $i$ elements in the arrival order. For sample-based models with sample $S$ of size $s$, we set $E_s := S$, and for $i\geq s$ we let $E_i$ be $S$ together with the first $(i-s)$ arrivals in the selection phase. In sample-based models, $i$ indexes how many elements have been revealed so far (including the sample); it is not a position in the preliminary adversarial order. In particular, if $E_i=\emptyset$, then $\tau(E_i)=0$, and since all weights are positive, every unseen element is counted.
\end{definition}

\paragraph{Assumption (reliability of the oracle).}
Throughout the paper, we assume that the oracle answers each query correctly. This baseline lets us isolate the value of ordinal query access from robustness issues; studying noisy or unreliable query responses is beyond our scope.

Note that $\mathrm{QR}_i$ is ordinal: it counts the unseen elements whose weight exceeds the current record, but it reveals no weights. In particular, $\mathrm{QR}_i=0$ if and only if the current record is the global maximum, while $\mathrm{QR}_i>0$ means that the maximum-weight element has not yet arrived.

We consider three query budgets: no access to the oracle, at most one query round, and an unlimited number of query rounds. For an arrival-order model $\mathcal{M}$, we write $\mathcal{M}$ for the version with no queries, $\mathcal{M}1$ for the version with one query round, and $\mathcal{M}\infty$ for the version with unlimited query rounds. Since we consider $4$ arrival-order models and $3$ query budgets, this yields $12$ models in total.

Whenever an algorithm below says that it returns the first future element satisfying a stated condition, it terminates without selecting an element if no such element appears.

\subsection{Adversarial Order}

Consider the Random Selection Algorithm below:

\begin{algorithm}[H]
    \DontPrintSemicolon
    $i \gets \text{Uniform}([n])$\;
    \Return element arriving in position $i$\;
	\caption{Random Selection Algorithm}
	\label{alg:zero}
\end{algorithm}
This algorithm is $1/n$-competitive for the AO model without queries. A folklore result states that this algorithm is optimal for that model. The following theorem extends this optimality to any number of query rounds, implying that $\alpha_{\mathrm{AO}} = \alpha_{\mathrm{AO}1} = \alpha_{\mathrm{AO}\infty} = 1/n$. The upper bound already holds for the oblivious AO model, since the proof uses only a distribution over fixed arrival orders.

\begin{theorem}
\label{theo:ao1}
    The random selection algorithm is $\frac{1}{n}$-competitive and optimal in the Adversarial Order model with Unlimited Query Rounds.
\end{theorem}

\begin{proof}
    To obtain the desired upper bound, we adopt the adversary's perspective. By Yao's principle, for any distribution $\mathcal{F}$ over the orders of $E$, we have
    $\alpha_{\mathrm{AO}\infty} \leq \max_{A} \PP_{\pi \sim\mathcal{F}}[A\;\wins]$,
    where the maximum is over all deterministic algorithms $A$. Thus, it suffices to exhibit a distribution $\mathcal{F}$ such that
    $\PP_{\pi\sim \mathcal{F}}[A\;\wins] \leq 1/n$
    for every deterministic algorithm $A$.

    Consider the orders $\pi_i$ for $i \in [n]$, defined such that $x^1$ arrives at position $i$, and before $x^1$ the worst $i-1$ elements arrive in increasing order of weight:
    \[
    \pi_i(j)=x^{n-j+1}\quad\text{for }j<i,
    \qquad
    \pi_i(i)=x^1.
    \]
    The remaining elements arrive after $x^1$ in any fixed order.

    Distribute $i$ uniformly in $[n]$. We use here that algorithms are ordinal: before an element's weight is revealed, its label carries no information about its hidden weight. Fix a deterministic ordinal algorithm. If it selects $x^1$ in $\pi_i$, then it cannot select $x^1$ in any $\pi_k$ with $k\ne i$. If $k<i$, then the histories up to step $k$ are indistinguishable before $x^1$ arrives in $\pi_k$. If $k>i$, then the histories up to step $i$ are indistinguishable, but the algorithm stops at step $i$ in $\pi_k$ when the arriving element is not $x^1$. Therefore, every deterministic algorithm wins with probability at most $1/n$.
    
\end{proof}

\subsection{Sample-Based Online Adversarial Order} \label{sec:onAOS}

First, we briefly discuss the SBOnAO model without queries. This model, also known in the literature as the \emph{order-oblivious secretary problem} (a term introduced by~\cite{azar2014prophet}), allows an algorithm to choose a sample size $s$, observe the values of a uniformly random subset of $s$ elements (rejecting all of them), and then face an adversary who decides the order in which the remaining elements are presented. A well-known $1/4$-competitive algorithm samples half of the elements and then selects the first element in the second phase that is better than anything seen so far. The factor $1/4$ is known to be optimal (see, e.g.,~\cite{Kaplan2022OnlineSample}). In Section~\ref{sec:SBOffAO}, we provide a stronger proof of the $1/4$ upper bound by switching to the SBOffAO model. We show that even against the weaker offline adversary, no algorithm can achieve a competitive ratio greater than $1/4$; since this adversary model is weaker than SBOnAO, the same upper bound applies to SBOnAO as well.

\subsubsection{One and Unlimited Query Rounds}

We now turn to the SBOnAO model with query access. Consider the algorithm below.

\begin{algorithm}[H]
    \DontPrintSemicolon
	\KwSty{Parameters:} $p\in[0,1]$\;
    $s \gets \Bin(n,p)$\;
    \KwSty{Sampling phase:} reject a uniform random sample $S\subseteq E$ of size $s$\;
    $\tau \gets \tau(S)$\;
    receive query response $r$ from the oracle\;
    \lIf{$r=0$}{terminate without selecting an element}
    $i \gets \Unif([r])$\;
    \KwSty{Selection phase:} \Return the $i$-th arriving element $x$ with $w(x)>\tau$\;
	\caption{Sample and Random Record Selection Algorithm}
	\label{alg:on1}
\end{algorithm}
\begin{theorem}
\label{theo:on1}
    Setting $p = \frac{1}{e}$, Algorithm~\ref{alg:on1} is $\frac{1}{e}$-competitive and optimal in the Sample-Based Online Adversarial Order model with Unlimited Query Rounds.
\end{theorem}

Since Algorithm~\ref{alg:on1} makes only one call to the oracle, the same result applies to the SBOnAO model with One Query Round. Thus, $\alpha_{\mathrm{SBOnAO}1} = \alpha_{\mathrm{SBOnAO}\infty} = 1/e$.

\begin{proof}
    \textbf{Lower bound.} We show that for any $p\in(0,1]$, Algorithm~\ref{alg:on1} is $(-p\ln p)$-competitive (and taking $p=1/e$ gives $1/e$).
    Let $r$ denote the number of post-sample elements that are heavier than every sampled element. Since $s\sim\Bin(n,p)$, each element is sampled independently with probability $p$, so
    \begin{align*}
        \PP[r=k]=
        \begin{cases}
            p(1-p)^k & \text{if } k \in \{0,1,\dots,n-1\}, \\
            (1-p)^n & \text{if } k = n.
        \end{cases}
    \end{align*}
    If $r=0$, the algorithm cannot select $x^1$. Otherwise, it selects a uniformly random element from $B := \{x\in E\setminus S : w(x)>\tau(S)\}$. Hence
    \begin{align*}
        \PP[A\;\wins\mid r=k]=\frac{1}{k}, \qquad \text{for } k\in[n].
    \end{align*}
    Therefore,
    \begin{align*}
        \PP[A\;\wins]
        &= \sum_{k=1}^n \PP[A\;\wins\mid r=k]\,\PP[r=k]
        = p\sum_{k=1}^{n-1} \frac{(1-p)^k}{k} + \frac{(1-p)^n}{n} \\
        &\ge p\sum_{k=1}^\infty \frac{(1-p)^k}{k} = -p\ln(p),
    \end{align*}
    where we use $\frac{(1-p)^n}{n} \ge p\sum_{k=n}^\infty \frac{(1-p)^k}{k}$ and $-\ln(p)=\sum_{k=1}^{\infty} \frac{(1-p)^k}{k}$.
    Taking $p=1/e$ implies that the algorithm is $1/e$-competitive.

    \textbf{Optimality.} It remains to prove the matching upper bound $\alpha_{\mathrm{SBOnAO}\infty}\le 1/e$. We do so by converting any SBOnAO$\infty$ algorithm into a Random Order algorithm, contradicting the classical $1/e$ upper bound in RO.

    Let $A=(\mathcal{D},A')$ be an arbitrary algorithm for SBOnAO$\infty$, and let $\alpha := \alpha_{\mathrm{SBOnAO}\infty}(A)$. On an instance with $n$ elements, $A$ draws $s\sim\mathcal{D}_n$, observes a uniformly random sample $S\subseteq E$ of size $s$, and then runs $A'$ on the remaining elements in an adversarial order.

    Since we are proving an upper bound and the query budget is unlimited, we may assume that $A'$ queries immediately after seeing $S$ and learns $r:=\mathrm{QR}_s$. Let
    \[
        B=\{x\in E\setminus S : w(x) > \tau(S)\}
    \]
    be the set of post-sample elements that beat the sample maximum; by definition, $|B|=r$. Only elements in $B$ can possibly be the global maximum $x^1$. Elements outside $B$ have weight at most $\tau(S)$, so selecting them cannot be successful. Conditioning on $S$ and $r$, we may therefore focus only on arrivals from $B$. Any way in which $A'$ selects an element of $B$ yields a strategy for choosing the maximum among the $r$ elements of $B$ arriving in adversarial order, with the same success probability. By Theorem~\ref{theo:ao1}, this conditional success probability is at most $1/r$.

To apply Yao's principle conditionally on $S$, the adversary orders $B$ according to the hard distribution of Theorem~\ref{theo:ao1}, and places the remaining unseen elements arbitrarily. Since elements outside $B$ are never winning selections, they do not affect the conditional upper bound.

    Define $g(0):=0$ and $g(r):=1/r$ for $r\ge 1$. Now define a modified \emph{SBOnAO$\infty$} algorithm $\widehat A=(\mathcal D,\widehat A')$ that depends on the instance only through $r$: $\widehat A'$ queries to obtain $r=\mathrm{QR}_s$, terminates if $r=0$, and otherwise picks $t\sim\Unif([r])$ and selects the $t$-th post-sample arrival with weight exceeding all sampled weights. Equivalently, for $r\ge 1$ it selects a uniformly random element of $B$, so
\[
    \PP[\widehat A\;\wins\mid r]=g(r).
\]
By the preceding paragraph, $\PP[A\;\wins\mid r]\le g(r)$ for every $r$. Taking expectation over the randomness of the sample (and hence of $r$), and using $\PP[\widehat A\;\wins\mid r]=g(r)$, we have
\[
    \PP[\widehat A\;\wins]=\mathbb E[\PP[\widehat A\;\wins\mid r]]=\mathbb E[g(r)]\ge \mathbb E[\PP[A\;\wins\mid r]]=\PP[A\;\wins]\ge \alpha.
\]

    Finally, define an RO algorithm $\widetilde A$ with no queries: draw $s\sim\mathcal D_n$, reject the first $s$ arrivals, set $\tau\gets\tau(E_s)$, and then select the first subsequent element with weight strictly larger than $\tau$. Conditioned on $r=\mathrm{QR}_s$, there are exactly $r$ above-threshold arrivals, their relative order is uniform in RO, and the first one is the maximum among them with probability $1/r$; thus $\PP[\widetilde A\;\wins\mid r]=g(r)$ and $\PP[\widetilde A\;\wins]=\mathbb E[g(r)]$. Therefore $\alpha\le \alpha_{\mathrm{RO}}(\widetilde A)\le 1/e$, where the last inequality uses the classical optimality of $1/e$ in RO.
\end{proof}

\subsection{Sample-Based Offline Adversarial Order} \label{sec:SBOffAO}
In this section, we present three results, one for each query budget.

The offline adversary is weaker than the online adversary because it fixes the preliminary order before the sample is realized. The upper bounds in this section therefore use distributions over preliminary orders. These distributions are known to the algorithm in the Yao argument, so proving an upper bound against them is enough to upper bound the original adversarial model.

\subsubsection{No Queries}

For the no-query case, consider the following threshold algorithm (a variant of Dynkin's classical algorithm).

\begin{algorithm}[H]
    \DontPrintSemicolon
	\KwSty{Parameters:} $p\in[0,1]$\;
    $s \gets \Bin(n,p)$\;
    \KwSty{Sampling phase:} reject a uniform random sample $S\subseteq E$ of size $s$\;
    $\tau \gets \tau(S)$\tcp*{maximum weight in sample}
    \KwSty{Selection phase:} \Return first $x$ such that $w(x)>\tau$\;
	\caption{Dynkin's Algorithm}
	\label{alg:dynkin}
\end{algorithm}
\begin{theorem}
\label{theo:off}
    Setting $p = \frac{1}{2}$, Algorithm~\ref{alg:dynkin} is $\frac{1}{4}$-competitive and optimal in the Sample-Based Offline Adversarial Order model without queries.
\end{theorem}

\begin{proof}
    \textbf{Lower bound.} For $p=1/2$, Algorithm~\ref{alg:dynkin} chooses a sample of size $s \sim \Bin(n,\frac{1}{2})$ and then selects the first record outside the sample. If $x^1$ is outside the sample and $x^2$ is inside the sample, then the algorithm selects $x^1$. Since this event occurs with probability $\frac{1}{4}$, the algorithm is $\frac{1}{4}$-competitive.

    \textbf{Optimality.} To show that no algorithm can select $x^1$ with probability greater than $\frac{1}{4}+o(1)$, we consider a distribution over preliminary adversary orders and apply Yao's principle. For $i\in[n]$, let $\psi_i\colon [n] \to E$ be any order satisfying $\psi_i(j)=x^{i-j+1}$ for all $j\leq i$. That is, $\psi_i$ presents the top $i$ elements in increasing order of weight $(x^i, x^{i-1}, \dots, x^1)$, and then presents the remaining elements in an arbitrary order.

    Let $\mathcal{F}$ be the distribution that selects one of the orders $\psi_i$ uniformly at random. The cases $s=0$ and $s=n$ are immediate: under $\mathcal{F}$, their success probabilities are at most $1/n$ and $0$, respectively. Hence fix a sample size $s\in\{1,\dots,n-1\}$. For intuition, imagine the elements placed from left to right according to an order $\psi\sim\mathcal{F}$, with exactly $s$ of them marked as sampled. Let $\bar{x}$ be the leftmost unsampled element that is heavier than every sampled element (equivalently, the first record outside the sample from the algorithm's perspective). Let $J$ be the random variable defined as the position $J=\psi^{-1}(\bar{x})$ of $\bar{x}$ if $\bar{x}$ exists, and $J=n+1$ otherwise.

Consider any deterministic ordinal algorithm $A$ for SBOffAO without queries that uses sample size $s$. We may give $A$ the value of $J$ after the sample is realized; this can only increase its success probability. We also may assume that $A$ never selects a non-record, since a non-record cannot be $x^1$ and rejecting it preserves all future options. Hence the first element at which $A$ can possibly stop is~$\bar{x}$.

Condition on $J=j$ and on the event that $x^1$ is not sampled. From $\bar{x}$ onward, the algorithm sees a sequence of records with strictly increasing weights, ending at $x^1$. If the preliminary order is $\psi_i$, then this record sequence has length $\ell=i-j+1$. Since the algorithm is deterministic and ordinal, while this increasing record sequence is being observed it has no information other than the number of records already seen and the value of $J=j$. Thus, conditional on $J=j$, its successful stopping rule is equivalent to choosing a position $t\ge 1$ in this record sequence, possibly not stopping at all. It wins only if $t=\ell$.

For fixed $j$, choosing $t$ is equivalent to choosing $i=j+t-1$. Therefore the best possible success probability conditional on $J=j$ is
\[
    \max_{i\ge j} \PP[\psi=\psi_i\mid J=j].
\]
By Bayes' rule and the uniform prior on $\{\psi_i\}_{i\in[n]}$, it is enough to maximize $\PP[J=j\mid \psi=\psi_i]$. If $i<j$, this probability is zero. If $i\ge j$ and $j\ne 1$, the event $J=j$ is exactly the event that $x^{i-j+2}\in S$ and $\{x^1,\dots,x^{i-j+1}\}\cap S=\emptyset$, and hence
\[
    \PP[J=j\mid\psi=\psi_i]=\binom{n-(i-j+2)}{s-1}/\binom{n}{s}.
\]
If $j=1$, the event is $\{x^1,\dots,x^i\}\cap S=\emptyset$, and
\[
    \PP[J=1\mid\psi=\psi_i]=\binom{n-i}{s}/\binom{n}{s}.
\]
In both cases the numerator decreases with $i$. Hence the posterior probability is maximized at the smallest feasible value, namely $i=j$, which corresponds to $t=1$.

Thus, even after giving the algorithm the extra information $J$, the best deterministic ordinal strategy is to select the first record outside the sample. Let $A^*_s$ be the algorithm that takes the same fixed sample size $s$ and returns the first record outside the sample. Then $\PP[A^*_s\;\wins]\ge \PP[A\;\wins]$. We compute its winning probability.

    If $\psi=\psi_1$, then $A^*_s$ wins if and only if $x^1$ is outside the sample. If $\psi\neq\psi_1$, then $A^*_s$ wins if and only if $x^1$ is outside the sample and $x^2$ is inside the sample (since for all $i\ge 2$, the order $\psi_i$ presents $x^2$ before $x^1$). Therefore,
    \[
        \PP[A^*_s\;\wins] = \frac{1}{n}\PP[A^*_s\;\wins\mid\psi=\psi_1] + \frac{n-1}{n}\PP[A^*_s\;\wins\mid\psi\neq\psi_1] = \frac{(n-s)}{n^2} + \frac{(n-s)s}{n^2}.
    \]

    If the algorithm randomizes the sample size, its success probability is a convex combination of the fixed-$s$ success probabilities, and hence is at most $\max_s \PP[A^*_s\;\wins]$. The maximum of the expression above over $s\in \{0,1,\dots,n\}$ is attained at $s = \lfloor\frac{n}{2}\rfloor$. Hence,
    \begin{align*}
        \alpha_{\text{SBOffAO}}
        &\leq \max_s \PP[A^*_s\;\wins] \\
        &\leq \frac{1}{n^2}\left\lceil\frac{n}{2}\right\rceil
          + \frac{1}{n^2}\left\lceil\frac{n}{2}\right\rceil
            \left\lfloor\frac{n}{2}\right\rfloor
         = \frac{1}{4}+o(1).\qedhere
    \end{align*}
\end{proof}

\subsubsection{One Query Round}

Next, we study the Sample-Based Offline Adversarial Order model with One Query Round (SBOffAO1). In Section~\ref{sec:onAOS}, we presented Algorithm~\ref{alg:on1} for the online-adversary model SBOnAO1. Since Algorithm~\ref{alg:on1} is $1/e$-competitive in SBOnAO1, it is also $1/e$-competitive in SBOffAO1 (the offline-adversary model is \emph{easier} for the algorithm, since any offline adversary induces a valid online adversary).

We conjecture that the optimal competitive ratio in SBOffAO1 is also $1/e$; if true, Algorithm~\ref{alg:on1} would be optimal for SBOffAO1. We prove this conjecture under the above-threshold arrivals restriction.

The restriction below is not needed by the algorithm; it is needed for the matching upper bound. Without it, sub-threshold arrivals may leak information about the preliminary order fixed by the offline adversary. This leakage is exactly what prevents the current argument from proving an upper bound for unrestricted SBOffAO$1$. Thus the unrestricted one-query case remains open.

\paragraph{Above-threshold arrivals.}
We restrict the post-sample information available to the algorithm as follows. After observing the sample $S$ and receiving the query response $r$, the algorithm sets the threshold $T:=\tau(S)$. From that point on, the algorithm only observes arrivals whose weight exceeds $T$; arrivals with weight at most $T$ are completely unobserved (in particular, the algorithm does not learn that they occurred, nor how many of them occurred between two observed arrivals). Equivalently, the online phase is revealed to the algorithm only through the subsequence of elements with $w(\cdot)>T$, starting at the first such element.
To simplify the upper-bound proof, we switch from a fixed-size uniform sample to i.i.d.\ sampling. The next lemma shows that, for the hard distribution used in the proof of Theorem~\ref{theo:off1indist}, this changes the resulting bound only by $o(1)$.
\begin{lemma}
    Let $s_n/n\to p$. Consider the distribution over preliminary orders used in the proof of Theorem~\ref{theo:off1indist}. For this distribution, the resulting upper bound on the success probability under a uniform sample of size $s_n$ differs by $o(1)$ from the corresponding upper bound obtained under independent sampling with probability $p$.
\end{lemma}
\begin{proof}
    For each $r\le K$, the probability of obtaining exactly $r$ candidates above the threshold under independent sampling is $p(1-p)^r$. Under uniform sampling of size $s_n$, it is $\frac{s_n}{n}\frac{(n-s_n)_{\underline{r}}}{(n-1)_{\underline{r}}}$, where $(a)_{\underline{r}}:=a(a-1)\cdots(a-r+1)$ denotes the falling factorial.

    We use only two features of the distribution: (i) the preliminary order is sampled independently of the sample, and (ii) conditional on $\mathrm{QR}_s=r$, the location of $x^1$ inside the observed above-threshold block is uniform.

    Therefore, before reaching $x^1$, every observed above-threshold element is a new record. No deterministic ordinal strategy can distinguish the true $x^1$ from these earlier arrivals, so the conditional success probability is at most $1/r$.

    The difference between the two sampling probabilities tends to zero as $n\to\infty$ for any fixed $K$. Since the conditional success probability given $r$ is at most $1/r$, the contribution from cases $r\le K$ differs by $o(1)$. For $r>K$, the contribution to the success probability is bounded by $\sum_{r>K} \PP[r]\,\frac{1}{r} \le \frac{1}{K}$. Taking $K$ large enough makes this term arbitrarily small. Thus the two bounds differ by $o(1)$.
\end{proof}

\begin{theorem}
\label{theo:off1indist}
Under the above-threshold arrivals restriction, Algorithm~\ref{alg:on1} is optimal within the restricted version of SBOffAO$1$, and the optimal value there is $1/e$.
\end{theorem}

\begin{proof}
As discussed above, Algorithm~\ref{alg:on1} achieves $1/e$ in SBOnAO1 and hence also in SBOffAO1. It therefore suffices to prove the matching upper bound for the restricted model under the above-threshold arrivals restriction.

To apply Yao's principle, we may restrict attention to deterministic algorithms with a fixed sample size $s_0$. For convenience, we instead use an equivalent coin-flip sampling model: for each element, independently flip a coin that comes up heads with probability $p:=s_0/n$, and include the element in the sample if and only if the coin is heads. Thus $s\sim\mathrm{Bin}(n,p)$, as justified by the lemma above.

Fix $p\in[0,1]$, let $s\sim\mathrm{Bin}(n,p)$, and let $S$ be the sampled set. After receiving $S$, the algorithm makes its unique oracle call and receives $r=\mathrm{QR}_s$. Fix a deterministic ordinal algorithm $A$ for SBOffAO1 satisfying the above-threshold arrivals restriction. We bound $\PP[A\;\wins]$ by exhibiting a distribution $\calF$ over preliminary orders $\psi:[n]\to E$.
\medskip

We define $\calF$ by an explicit random construction of $\psi$. Write $E=\{x^1,x^2,\dots,x^n\}$ with $w(x^1)>\cdots>w(x^n)$. For each $t\in[n]$, we maintain an ordered list
\[
B_t=\bigl(b^{(t)}_1,b^{(t)}_2,\dots,b^{(t)}_t\bigr)
\qquad\text{with underlying set }\{x^1,\dots,x^t\}.
\]
Let $K_t\in[t]$ be the index such that $b^{(t)}_{K_t}=x^1$. Initialize $B_1=(x^1)$ and $K_1=1$.

For $t=2,3,\dots,n$, given $(B_{t-1},K_{t-1})$, define $B_t$ by inserting $x^t$ at one of the two ends. With probability $K_{t-1}/t$, set $B_t=(x^t,b^{(t-1)}_1,\dots,b^{(t-1)}_{t-1})$ and $K_t=K_{t-1}+1$. Otherwise, set $B_t=(b^{(t-1)}_1,\dots,b^{(t-1)}_{t-1},x^t)$ and $K_t=K_{t-1}$. Finally, define the preliminary order $\psi$ by $\psi(j):=b^{(n)}_j$ for $j\in[n]$.

Two features of this construction are what we need. First, for every $t$, the set $\{x^1,\dots,x^t\}$ occupies a contiguous interval in $\psi$ (later steps only add new elements at the ends of the current block). Moreover, within that interval the weights strictly increase up to $x^1$ and then strictly decrease, since each new element $x^t$ is worse than all previous ones. Second, $K_t$ is uniform over $[t]$ for every $t$. We prove this claim by induction below.

The claim is immediate for $t=1$. Fix $t\ge 2$ and assume $K_{t-1}$ is uniform on $[t-1]$. For any $i\in\{2,\dots,t-1\}$, the event $\{K_t=i\}$ occurs in exactly two disjoint ways: either $K_{t-1}=i$ and we insert $x^t$ on the right (probability $1-i/t$), or $K_{t-1}=i-1$ and we insert $x^t$ on the left (probability $(i-1)/t$). Hence
\[
\PP[K_t=i]=\PP[K_{t-1}=i]\Bigl(1-\frac{i}{t}\Bigr)+\PP[K_{t-1}=i-1]\Bigl(\frac{i-1}{t}\Bigr)=\frac{1}{t-1}\Bigl(1-\frac{i}{t}\Bigr)+\frac{1}{t-1}\Bigl(\frac{i-1}{t}\Bigr)=\frac{1}{t}.
\]

We now complete the induction by handling the boundary cases. The boundary cases are similar and use only one of the two possibilities. For $i=1$, inserting on the left would shift $x^1$ to position $2$, so
\[
\PP[K_t=1]=\PP[K_{t-1}=1](1-1/t)=\frac{1}{t-1}\cdot\frac{t-1}{t}=1/t.
\]
For $i=t$, insertion on the right keeps $x^1$ within the first $t-1$ positions, so
\[
\PP[K_t=t]=\PP[K_{t-1}=t-1]\cdot\frac{t-1}{t}=\frac{1}{t-1}\cdot\frac{t-1}{t}=1/t.
\]
Thus $K_t$ is uniform on $\{1,\dots,t\}$ for all $t$.

Let us continue with the proof of the theorem. Condition on the event $r=k$ with $1\le k\le n-1$. By definition of $r$, the maximum element in the sample is $x^{k+1}$, and the elements of $E\setminus S$ beating it are exactly $\{x^1,\dots,x^k\}$. Under i.i.d.\ sampling,
\[
\PP[r=k]=p(1-p)^k\quad (k\le n-1),
\qquad\text{and}\qquad
\PP[r=n]=(1-p)^n.
\]

Since the construction makes $\{x^1,\dots,x^k\}$ a contiguous block in $\psi$, and none of these elements is sampled when $r=k$, the same $k$ elements appear as a contiguous block in the induced order $\phi$ seen after the sample.

Under the above-threshold arrivals restriction, after receiving the sample $S$ and the query response $r$, the algorithm observes only the post-sample subsequence whose weights exceed $\tau(S)=w(x^{k+1})$, starting from the first such arrival. Conditioned on $r=k$, this observed subsequence is exactly the length-$k$ block consisting of $\{x^1,\dots,x^k\}$.

Let $J\in[k]$ be the position of $x^1$ inside this observed block. By uniformity of $K_k$ and independence of sampling from $\psi$,
\[
\PP[J=j\mid r=k]=\frac{1}{k}\qquad\text{for all }j\in[k].
\]
We claim that, conditional on $r=k$, no deterministic ordinal algorithm satisfying the above-threshold arrivals restriction can win with probability more than $1/k$.

Indeed, before the observed process reaches $x^1$, every revealed above-threshold element is a new record. Hence, upon seeing the $t$-th such arrival, the algorithm cannot distinguish whether $J=t$ or $J>t$.

Thus any (possibly randomized) strategy is described by a subprobability vector $(a_1,\dots,a_k)$. Here $a_t$ is the probability that the algorithm stops at the $t$-th observed above-threshold arrival, and $\sum_{t=1}^k a_t\le 1$. Since $J$ is uniform on $[k]$,
\[
\PP[A\;\wins\mid r=k]
=\sum_{t=1}^k a_t\PP[J=t\mid r=k]
=\frac{1}{k}\sum_{t=1}^k a_t
\le \frac{1}{k}.
\]

Therefore,
\[
\PP[A\;\wins]
\le \sum_{k=1}^{n-1}\PP[r=k]\cdot\frac{1}{k}+\PP[r=n]\cdot\frac{1}{n}
= p\sum_{k=1}^{n-1}\frac{(1-p)^k}{k}+\frac{(1-p)^n}{n}.
\]

If $p=0$, the bound above gives $\PP[A\;\wins]\le 1/n$. This implies the desired asymptotic upper bound. We therefore assume $p\in(0,1]$ for the remaining calculation.

Using $\sum_{k\ge 1}(1-p)^k/k=-\ln(p)$ and $\sum_{k\ge n}(1-p)^k/k\ge (1-p)^n/n$, we obtain
\[
p\sum_{k=1}^{n-1}\frac{(1-p)^k}{k}+\frac{(1-p)^n}{n}
= -p\ln(p)-p\sum_{k\ge n}\frac{(1-p)^k}{k}+\frac{(1-p)^n}{n}
\le -p\ln(p)+\frac{(1-p)^{n+1}}{n}.
\]
Hence,
\[
\PP[A\;\wins]\le -p\ln(p)+\frac{(1-p)^{n+1}}{n}.
\]
Since $\max_{p\in(0,1]}(-p\ln p)=1/e$, we conclude that $\PP[A\;\wins]\le 1/e+o(1)$, completing the proof.
\qedhere
\end{proof}

\begin{remark}[Small examples for the hard distribution]
To build intuition for the construction of $\calF$ in the proof of Theorem~\ref{theo:off1indist}, we list the resulting distribution explicitly for $n=3,4,5$. We write $(i_1,\dots,i_n)$ to denote the order $(x^{i_1},\dots,x^{i_n})$.

\smallskip
\noindent The distribution $\calF$ for $n=3$ is as follows:
\[
(3,2,1)\ \text{and}\ (1,2,3)\ \text{have probability } \frac{1}{3},\qquad
(2,1,3)\ \text{and}\ (3,1,2)\ \text{have probability } \frac{1}{6}.
\]

\smallskip
\noindent For $n=4$, the distribution is:
\[
(1,2,3,4)\ \text{and}\ (4,3,2,1)\ \text{have probability } \frac{1}{4},
\]
and each of
\[
(3,2,1,4),\ (4,2,1,3),\ (2,1,3,4),\ (4,3,1,2),\ (3,1,2,4),\ (4,1,2,3)
\]
has probability $1/12$.

\smallskip
\noindent For $n=5$, the distribution is:
\[
(1,2,3,4,5)\ \text{and}\ (5,4,3,2,1)\ \text{have probability } \frac{1}{5},
\]
each of
\begin{gather*}
(2,1,3,4,5),\ (3,1,2,4,5),\ (4,1,2,3,5),\ (5,1,2,3,4),\ \
(4,3,2,1,5),\ (5,3,2,1,4),\ (5,4,2,1,3),\ (5,4,3,1,2)
\end{gather*}
has probability $1/20$, and each of
\begin{gather*}
(3,2,1,4,5),\ (4,2,1,3,5),\ (4,3,1,2,5),\ \
(5,2,1,3,4),\ (5,3,1,2,4),\ (5,4,1,2,3)
\end{gather*}
has probability $1/30$.
\end{remark}

Without the \emph{above-threshold arrivals restriction}, we provide the following slightly weaker upper bound in a restricted case.

\begin{theorem}
\label{theo:off1}
    Let $A$ be an algorithm for the Sample-Based Offline Adversarial Order model with One Query Round (SBOffAO1) that uses its oracle call immediately after receiving the sample. Then the competitive ratio of $A$ is less than $0.3787$.
\end{theorem}
\begin{proof}
As in previous proofs, we consider a distribution $\mathcal{F}$ over preliminary adversary orders. We show that any deterministic ordinal algorithm that calls the oracle immediately after receiving the sample selects $x^1$ with probability less than $0.3787$.

To apply Yao's principle, it is enough to bound deterministic post-sample algorithms for each fixed sample size $s$. We first prove such a fixed-size bound, and then we average over an arbitrary distribution $D_n$ over sample sizes.

Suppose that for large $n\in\NN$ the adversary draws one of the following preliminary orders with the indicated probabilities. For brevity, we use the notation $\psi=abcd\#$ to mean that $\psi(1)=a$, $\psi(2)=b$, $\psi(3)=c$, $\psi(4)=d$, and the remaining elements appear in decreasing order of weight:
\begin{align*}
    \psi_1 &= x^1x^2x^3x^4\# \text{ with probability }\frac{1}{4}, &\psi_2 = x^4x^3x^2x^1\# \text{ with probability }\frac{1}{4},\\
    \psi_3 &= x^2x^1x^3x^4\# \text{ with probability }\frac{1}{6}, &\psi_4 = x^4x^3x^1x^2\# \text{ with probability }\frac{1}{6},\\
    \psi_5 &= x^3x^2x^1x^4\# \text{ with probability }\frac{1}{12}, &\psi_6 = x^4x^1x^2x^3\# \text{ with probability }\frac{1}{12}.
\end{align*}

Fix $s\in\{0,\dots,n\}$, let $S$ be a uniform size-$s$ sample, and let $r=\mathrm{QR}_s$. We bound $\PP[A\;\wins\mid r]$ under the above distribution over $\psi$ as follows.
\[
\PP[A\;\wins\mid r=k] \le \frac{1}{k}\quad\text{for }k\in\{2,3\},
\qquad\text{and}\qquad
\PP[A\;\wins\mid r\ge 4] \le \frac{1}{4}.
\]

Indeed, $r\ge 4$ ensures that none of $x^1, \dots, x^4$ is sampled. Since the suffix $\#$ is identical across all six preliminary orders, the subset of sampled elements below $x^4$ has the exact same distribution in all cases.

When $r\ge 4$, the oracle call provides no additional information about the preliminary order. In particular, for every $i\in[6]$, the conditional probability that $\psi=\psi_i$, given everything the algorithm knows before seeing the first post-sample element, equals its unconditional probability $\PP(\psi=\psi_i)$.

Moreover, up to the arrival of $x^1$, the algorithm sees the same ordinal information in all six cases: each new element is larger than all previously seen elements. The first difference occurs only when the element after $x^1$ arrives. Since $\PP[\psi^{-1}(x^1)=j]=\frac{1}{4}$ for all $j\in[4]$, it follows (by the same argument as in the proof of Theorem~\ref{theo:ao1}) that $\PP[A\;\wins\mid r\ge 4]\le \frac{1}{4}$.

When $r=3$, we must have $x^4\in S$ and $\{x^1,x^2,x^3\}\cap S=\emptyset$. Let $\tilde{\psi}$ denote the order obtained after removing the sample. Then
\begin{align*}
    \tilde{\psi}_1 &= x^1x^2x^3\# \text{ with probability } \frac{1}{3}, &\tilde{\psi}_2 = x^3x^2x^1\# \text{ with probability } \frac{1}{3},\\
    \tilde{\psi}_3 &= x^2x^1x^3\# \text{ with probability } \frac{1}{6}, &\tilde{\psi}_4 = x^3x^1x^2\# \text{ with probability } \frac{1}{6}.
\end{align*}
Again, all elements before $x^1$ appear in increasing order of weight, and $\PP[\tilde{\psi}^{-1}(x^1)=i]=\frac{1}{3}$ for all $i\in[3]$. Hence $\PP[A\;\wins\mid r=3]\le \frac{1}{3}$.

The case $r=2$ is slightly more delicate. Here, $r=2$ implies $x^3\in S$ and $\{x^1,x^2\}\cap S=\emptyset$, but $x^4$ may or may not be in $S$. Let $C$ be the event that the first post-sample element is a record.

Since the algorithm's behavior may depend on whether $C$ occurs, it suffices to show
\[
\PP[A\;\wins\mid r=2,C]\le \frac{1}{2}
\qquad\text{and}\qquad
\PP[A\;\wins\mid r=2,\bar{C}]\le \frac{1}{2}.
\]
Equivalently, it suffices to show that
\[
\PP\bigl[\psi^{-1}(x^1)<\psi^{-1}(x^2)\mid r=2,C\bigr]
=\PP\bigl[\psi^{-1}(x^1)<\psi^{-1}(x^2)\mid r=2,\bar{C}\bigr]
=\frac{1}{2}.
\]
Conditioned on $r=2$, we have $x^3\in S$ and $x^1,x^2\notin S$. If $x^4\in S$, then $C$ always occurs and
\[
\PP[\psi^{-1}(x^1)<\psi^{-1}(x^2)\mid r=2,C,x^4\in S]
=\frac{1/4+1/6+1/12}{1}=\frac12,
\]
since the favorable orders are $\psi_1,\psi_4,\psi_6$.

If $x^4\notin S$ and $C$ occurs, then necessarily $\psi\in\{\psi_1,\psi_3,\psi_5\}$, and only $\psi_1$ is favorable. Hence
\[
\PP\bigl[\psi^{-1}(x^1)<\psi^{-1}(x^2)\mid r=2,C,x^4\notin S\bigr]
=\frac{1/4}{1/4+1/6+1/12}=\frac12.
\]
If $x^4\notin S$ and $\bar C$ occurs, then necessarily $\psi\in\{\psi_2,\psi_4,\psi_6\}$, and $\psi_4,\psi_6$ are favorable. Thus
\[
\PP[\psi^{-1}(x^1)<\psi^{-1}(x^2)\mid r=2,\bar C]
=\frac{1/6+1/12}{1/4+1/6+1/12}=\frac12.
\]
We conclude that $\PP[A\;\wins\mid r=2]\le \frac{1}{2}$.

For fixed $s$,
\[
\PP[A\text{ wins}\mid s]\le F_n(s),
\]
where
\[
F_n(s):=
\PP[r=1\mid s]
+\frac12\PP[r=2\mid s]
+\frac13\PP[r=3\mid s]
+\frac14\PP[r\ge4\mid s].
\]

Let $p=s/n$. Since the expressions above contain only products of length at most four, the approximation error is $o(1)$ uniformly over $s\in\{0,\ldots,n\}$. Moreover,
\[
F_n(s)
=
p(1-p)+\frac12p(1-p)^2+\frac13p(1-p)^3+\frac14(1-p)^4+o(1).
\]

Now let $D_n$ be any distribution over sample sizes. Averaging the fixed-size bound gives
\[
\PP[A\text{ wins}]
\le \mathbb E_{s\sim D_n}[F_n(s)]
\le \max_{p\in[0,1]}
\left(
p(1-p)+\frac12p(1-p)^2+\frac13p(1-p)^3+\frac14(1-p)^4
\right)+o(1).
\]

The maximum is less than $0.3787$. By Yao's principle, the same upper bound applies to randomized immediate-query algorithms.
\end{proof}

\paragraph{Remark.}
The adversarial distribution used in the proof of Theorem~\ref{theo:off1indist} is tailored to the above-threshold arrivals restriction and need not be hard for unrestricted algorithms. Under this distribution, arrivals below the threshold may be correlated with the position of $x^1$ within the block $\{x^1,\dots,x^r\}$, and thus non-records can leak information.

For example, condition on $r=2$ and on the event that $x^4\notin S$. If the algorithm is allowed to observe the post-sample arrivals below the threshold, then before seeing the first record it may learn (even with some error) that $x^4$ has already arrived. This makes it more likely that the observed block is $(x^2,x^1)$ rather than $(x^1,x^2)$.

\subsubsection{Unlimited Query Rounds}
In the Unlimited Query Rounds version (SBOffAO$\infty$), the algorithm is more powerful than in the One Query Round version, where one can achieve a ratio of $1/e\approx 0.3678$. At the same time, SBOffAO$\infty$ is (in principle) harder than Random Order with Unlimited Query Rounds, where the competitive ratio is $1/2$ (Theorem~\ref{theo:rounl}). One might therefore expect the optimal competitive ratio for SBOffAO$\infty$ to coincide with one of these two extremes. In this section, we determine the exact competitive ratio for SBOffAO$\infty$ and show that it lies strictly between them: $\alpha_{\text{SBOffAO}\infty}=\frac{20\sqrt{10}-29}{81}\approx 0.4228.$

We depict the algorithm as Algorithm~\ref{alg:offunl}. In words, the algorithm first draws a random sample of size $s\sim\Bin(n,p)$, sets a threshold $\tau \gets \tau(S)$, and makes an oracle call to obtain the current query response $r$.

In the selection phase, the algorithm branches on $r$. If $r\in\{1,2\}$, it selects the first post-sample element exceeding $\tau$. If $r=3$, it either selects this first above-threshold element (with probability $q$) or rejects it and uses a second oracle call to decide between the next one or two candidates. If $r\ge 4$, it keeps querying while rejecting until the response drops below $3$, at which point it makes its final choice (using probability $h$ when $r=2$).

Let $A$ denote Algorithm~\ref{alg:offunl}.

\begin{algorithm}
    \SetAlgoNoLine
    \DontPrintSemicolon
	\KwSty{Parameters:} $p,q,h\in[0,1]$\;
    $s \gets \Bin(n,p)$\;
    \KwSty{Sampling phase:} reject a uniform random sample $S\subseteq E$ of size $s$\;
    $\tau \gets \tau(S)$ and receive query response $r$ from the oracle\;
    \KwSty{Selection phase:} \\
        \Indp \lCase{$r\in\{1,2\}$}{\Return first $x$ such that $w(x)>\tau$}
        \uCase{$r=3$}{
            \KwSty{with probability} $q$: \Return first $x$ such that $w(x)>\tau$\;
            \uElse{
            reject first $x$ such that $w(x)>\tau$ at step $t$\;
            recalculate $r \gets \mathrm{QR}_t$, $\tau \gets \tau(E_t)$\;
            \lIf{$r=0$}{terminate without selecting an element}
            \lElseIf{$r=1$}{\Return first $x$ such that $w(x)>\tau$}
            \lElse{\Return first or second $x$ such that $w(x)>\tau$ w.p. $\frac{1}{2}$ each}
            }
        }
        \uCase{$r\geq 4$}{
            \lWhile{$r\geq 3$ }{
                reject the next arriving element, say at time $i$, and update $r \gets \mathrm{QR}_i$, $\tau \gets \tau(E_i)$}
            \uIf{$r=2$}{
                \KwSty{with probability} h: \Return first $x$ such that $w(x)>\tau$\;
                \lElse{\Return second $x$ such that $w(x)>\tau$}
            }
            \lIf{$r=1$}{\Return first $x$ such that $w(x)>\tau$}
            \lIf{$r=0$}{terminate without selecting an element}
            }
	\caption{SBOffAO Algorithm with Unlimited Query Rounds}
	\label{alg:offunl}
\end{algorithm}
For later use in the proof of Theorem~\ref{theo:offunl}, we record a convenient case analysis for the win probabilities of $A$ on the relative order of $x^1,x^2,x^3$.

For distinct $a,b,c \in \{x^1,x^2,x^3\}$, let $((a,b,c))$ denote the set of preliminary orders $\psi$ such that $\psi^{-1}(a)<\psi^{-1}(b)<\psi^{-1}(c)$. Define
\[
    \eta(a,b,c) := \PP[A\; \wins\mid\psi\in ((a,b,c))].
\]

\begin{lemma}
\label{lem:offunl-case}
With the notation above, the following identities hold:
\begin{align*}
    \eta(x^1,x^2,x^3)&=\eta(x^1,x^3,x^2)=p(1-p)+p(1-p)^2+qp(1-p)^3,\\
    \eta(x^2,x^3,x^1)&=\eta(x^2,x^1,x^3)=p(1-p)+(1-q)p(1-p)^3+(1-p)^4,\\
    \eta(x^3,x^1,x^2)&=p(1-p)+p(1-p)^2+\frac{1}{2}(1-q)p(1-p)^3+h(1-p)^4,\\
    \eta(x^3,x^2,x^1)&=p(1-p)+\frac{1}{2}(1-q)p(1-p)^3+(1-h)(1-p)^4.
\end{align*}
\end{lemma}
\begin{proof}
Note that
\begin{align*}
    \eta(a,b,c) &= \sum_{i=1}^3 \PP[A\;\wins\mid \psi\in((a,b,c)),r_s=i] \PP[r_s=i] + \PP[A\;\wins\mid \psi\in((a,b,c)),r_s\geq 4] \PP[r_s\geq 4]\\
    &= \sum_{i=1}^3 \PP[A\;\wins\mid \psi\in((a,b,c)),r_s=i] \, p(1-p)^i + \PP[A\;\wins\mid \psi\in((a,b,c)),r_s\geq 4] \, (1-p)^4.
\end{align*}

The first line corresponds to the case where $x^1$ is the first element among the three heaviest. If $r_s=1$, then $A$ always wins. If $r_s=2$, then $A$ selects the first record it sees (namely $x^1$) and wins. If $r_s=3$, then $A$ selects the first record (again $x^1$) with probability $q$, and hence wins with probability $q$. If $r_s\ge 4$, then $A$ loses. In this case it keeps calling the oracle and does not select until it receives a response $r\le 2$, which can happen only after rejecting $x^1$. This proves the first line.

Next, consider the second line (when $x^2$ is first among the three heaviest and $x^1$ is second). If $r_s=1$, then $A$ wins. If $r_s=2$, then $A$ selects the first record it sees (which is $x^2$) and loses. If $r_s=3$, then $A$ wins if and only if it rejects the first record $x^2$. This happens with probability $(1-q)$. After that rejection, $A$ calls the oracle, receives $r=1$, and then wins. If $r_s\ge 4$, then $A$ wins. In this case it keeps calling the oracle and does not select until it receives $r\le 2$, which happens after rejecting $x^2$; at that point the oracle returns $r=1$ and $A$ wins. This proves the second line.

For the third line, $x^3$ is first, $x^1$ is second, and $x^2$ is third. If $r_s=1$, then $A$ wins since $x^1$ is selected. If $r_s=2$, then $A$ sees $x^3$ and selects it, and hence loses. If $r_s=3$, it rejects $x^3$ with probability $1-q$; the oracle then gives $r=2$. Facing $x^1$, it chooses it with probability $1/2$, resulting in $\frac{1}{2}(1-q)$. If $r_s\ge 4$, it queries until $r\le 2$, which happens after $x^3$. Then $r=2$, and $A$ chooses $x^1$ with probability $h$. This gives the third line.

For the fourth line, $x^3$ is first, $x^2$ is second, and $x^1$ is third. The analysis is identical except that after $x^3$, if $r=2$, $A$ wants to choose the second record ($x^1$) rather than the first ($x^2$). It does so with probability $1/2$ when $r_s=3$, and with probability $1-h$ when $r_s\ge 4$. This completes the proof.
\end{proof}

\begin{theorem}
\label{theo:offunl}
    Algorithm~\ref{alg:offunl} is $\frac{20\sqrt{10}-29}{81}$-competitive and optimal in the Sample-Based Offline Adversarial Order model with Unlimited Query Rounds.
\end{theorem}

\begin{proof}
    \textbf{Lower bound.} Assume $n\ge 4$, since for $n\le 3$ there is a $1/2$-competitive algorithm. Let $A$ be Algorithm~\ref{alg:offunl}. Choose parameters $p,q,h\in[0,1]$ and set $r_s:=\mathrm{QR}_s$ with $s\sim\mathrm{Bin}(n,p)$. Then $\PP[r_s=k]=p(1-p)^k$ for $k\in\{1,2,3\}$, and $\PP[r_s\ge 4]=(1-p)^4$.

    The algorithm's performance depends only on the relative order of the three heaviest elements. When $r_s\leq 3$, only the three heaviest can be selected. When $r_s\geq 4$, the algorithm keeps rejecting elements until the query response becomes $2$ or $1$. This can happen only after rejecting at least one element of $\{x^1,x^2,x^3\}$.

    By Lemma~\ref{lem:offunl-case}, the probability that $A$ wins against each $((a,b,c))$ is given by the following expressions:
    \begin{align*}
        \eta(x^1,x^2,x^3)&=\eta(x^1,x^3,x^2)=p(1-p)+p(1-p)^2+qp(1-p)^3,\\
        \eta(x^2,x^3,x^1)&=\eta(x^2,x^1,x^3)=p(1-p)+(1-q)p(1-p)^3+(1-p)^4,\\
        \eta(x^3,x^1,x^2)&=p(1-p)+p(1-p)^2+\frac{1}{2}(1-q)p(1-p)^3+h(1-p)^4,\\
        \eta(x^3,x^2,x^1)&=p(1-p)+\frac{1}{2}(1-q)p(1-p)^3+(1-h)(1-p)^4.
    \end{align*}

    By the definition of the competitive ratio, 
    \begin{align*}
        \alpha_{\mathrm{SBOffAO}\infty}(A)
        = \min\left(\eta(x^1,x^2,x^3),\eta(x^2,x^3,x^1),\eta(x^3,x^1,x^2),\eta(x^3,x^2,x^1)\right).
    \end{align*}
    Note that
\[
\eta(x^2,x^3,x^1)-\eta(x^3,x^2,x^1)
=\frac12(1-q)p(1-p)^3+h(1-p)^4\ge 0,
\]
for $q,h\in[0,1]$. Hence we may omit $\eta(x^2,x^3,x^1)$ from the minimum.

We then impose $\gamma=\eta(x^1,x^2,x^3)=\eta(x^3,x^1,x^2)=\eta(x^3,x^2,x^1)$. This yields the relations $q=\frac{1-2p}{3p(1-p)}$ and $h=\frac{p^2-3p+1}{2(1-p)^2}$.

Finally, we maximize the common expression $\gamma$ over $p\in[0,1]$ to obtain parameters $p = \frac{4-\sqrt{10}}{3}$, $q=\frac{10+\sqrt{10}}{18}$, and $h = \frac{1+\sqrt{10}}{18}$. Substituting these values yields $\alpha_{\mathrm{SBOffAO}\infty}(A) \geq \frac{20\sqrt{10}-29}{81}$, proving the lower bound.

    \textbf{Optimality.} For the upper bound, we consider a distribution over preliminary adversary orders. Fix an arbitrary order $\rho$ of $E\setminus\{x^1,x^2,x^3\}$, and define $\psi_1:=(x^1,x^2,x^3,\rho)$, $\psi_2:=(x^3,x^1,x^2,\rho)$, and $\psi_3:=(x^3,x^2,x^1,\rho)$. The adversary selects $\psi\in\{\psi_1,\psi_2,\psi_3\}$ uniformly at random. We claim that no algorithm can select $x^1$ with probability greater than $\frac{20\sqrt{10}-29}{81}$ against this distribution in the limit as $n\to\infty$.

    It suffices to bound deterministic algorithms with a fixed sample size $s=s(n)$ (randomizing $s$ cannot help). Let $r_s:=\mathrm{QR}_s$. We claim that, against the above distribution,
    \[
        \PP[A\;\wins\mid r_s\geq 3] \leq 1/3
        \qquad\text{and}\qquad
        \PP[A\;\wins\mid r_s=2]\leq 2/3,
    \]
    and trivially $\PP[A\;\wins\mid r_s=1]\leq 1$.

Conditioned on $r_s\ge 3$, the elements $x^1,x^2,x^3$ lie outside the sample and are heavier than every sampled element. Before the first of these three elements arrives, all observed elements are below the sample maximum and cannot be selected successfully. Thus the first relevant decision is made when the first element of $\{x^1,x^2,x^3\}$ appears. At that time the three preliminary orders remain equally likely. If the algorithm selects this element, it wins with probability $1/3$.

Suppose instead that the algorithm rejects the first relevant element. A later record appears only in the two orders whose first relevant element is not $x^1$. Conditional on this event, the next record is $x^1$ in exactly one of the two remaining orders. Hence the best possible success probability after rejecting the first relevant element is at most $(2/3)(1/2)=1/3$. Rejecting that later record as well cannot help, because in the only remaining winning case the algorithm has already passed $x^1$. Therefore $\PP[A\;\wins\mid r_s\ge 3]\le 1/3$.

Conditioned on $r_s=2$, the element $x^3$ is in the sample and $x^1,x^2$ are outside it. The first relevant post-sample element is a record. With probability $2/3$, the preliminary order is $\psi_1$ or $\psi_2$, and this first relevant element is $x^1$. With probability $1/3$, the preliminary order is $\psi_3$, and this first relevant element is $x^2$. Selecting the first relevant element wins with probability $2/3$. Rejecting it wins with probability at most $1/3$. Thus $\PP[A\;\wins\mid r_s=2]\le 2/3$.

    Therefore, for a fixed sample size $s$,
    \begin{align*}
        \PP[A\;\wins]
        &\leq 1\cdot\PP[r_s=1] + \frac{2}{3}\cdot\PP[r_s=2]+\frac{1}{3}\cdot\PP[r_s\geq 3] \\
        &=\frac{s(n-s)}{n(n-1)} + \frac{2}{3}\frac{s(n-s)(n-s-1)}{n(n-1)(n-2)}+\frac{1}{3}\frac{(n-s)(n-s-1)(n-s-2)}{n(n-1)(n-2)}.
    \end{align*}

    Let $u=(n-s)/n$. Along any convergent subsequence, the preceding upper bound tends to
    \[
        u(1-u)+\frac{2}{3}u^2(1-u)+\frac{1}{3}u^3 .
    \]
    The maximum of this cubic over $u\in[0,1]$ is attained at
    $u=(\sqrt{10}-1)/3$ and equals $(20\sqrt{10}-29)/81$.\qedhere

\end{proof}

\subsection{Random Order}

First, we discuss the Random Order model without queries. This is the classical secretary problem studied by Lindley~\cite{Lindley1961} and Dynkin~\cite{Dynkin1963}, and it has competitive ratio $1/e$, achieved by Dynkin's algorithm (Algorithm~\ref{alg:dynkin}) with $p=1/e$.

For a fixed number of candidates $n \in \NN^*$, the same algorithm with a suitable fixed sample size $s_n$ (instead of sampling $s\sim\Bin(n,p)$) has success probability $q_n$ that is strictly larger:
\begin{equation}
    q_n = \frac{1}{n}\sum_{i=s_n+1}^n \frac{s_n}{i-1} = \max_{s \in \{0,\dots,n-1\}} \frac{1}{n}\sum_{i=s+1}^n \frac{s}{i-1}.
\end{equation}

We use these values $s_n$ to define an algorithm for the Random Order model with One Query Round.

\subsubsection{One Query Round}

For $r\ge 1$, let $q_r$ be the optimal success probability in the classical secretary problem on $r$ candidates, and let $s_r$ be the corresponding optimal number of rejected candidates. Set $q_0:=0$.

For the One Query Round model, we propose Algorithm~\ref{alg:ro1}. The algorithm first takes a sample of size $s\sim\Bin(n,p)$ and then calls the oracle. The oracle returns $r=\mathrm{QR}_s$, so the algorithm learns that exactly $r$ unseen elements are heavier than all elements seen so far. It then runs an optimal strategy for the classical secretary problem with $r$ candidates, which succeeds with probability $q_r$ when $r>0$.

\begin{algorithm}

    \DontPrintSemicolon
	\KwSty{Parameters:} $p\in[0,1]$\;
    $s \gets \Bin(n,p)$\;
    \KwSty{1st sampling phase:} reject the first $s$ elements, and let $E_s$ be this prefix\;
    $\tau \gets \tau(E_s)$\;
    receive query response $r=\mathrm{QR}_s$ from the oracle\;
    \lIf{$r=0$}{terminate without selecting an element}
    \KwSty{2nd sampling phase:} reject the first $s_r$ subsequently arriving elements $x$ with $w(x)>\tau$\;
    let $E_t$ be the observed set after this second sampling phase, and set $\tau\gets\tau(E_t)$\;
    \KwSty{Selection phase:} \Return the first subsequently arriving element $x$ such that $w(x)>\tau$, if such an element appears\;
	\caption{Double Sample Algorithm}
	\label{alg:ro1}
\end{algorithm}
The upper bound for one query is slightly more delicate than the lower bound. It is not enough to analyze algorithms that query after a fixed prefix, because an arbitrary one-query algorithm could choose its query time from the ordinal history. The proof therefore first reduces arbitrary one-query algorithms to fixed-prefix query algorithms.

\begin{theorem}
\label{theo:ro1}
    Setting $p \approx 0.2029$, Algorithm~\ref{alg:ro1} is $0.4477\ldots$-competitive and optimal in the Random Order with One Query Round model.
\end{theorem}

\begin{proof}
\textbf{Lower bound.} We first analyze Algorithm~\ref{alg:ro1}. If the query returns $r>0$, then the $r$ elements above the threshold appear in a uniformly random relative order. Running the optimal $r$-candidate secretary rule on these elements succeeds with probability $q_r$. Thus the success probability of Algorithm~\ref{alg:ro1} is
\[
    p\sum_{r=1}^{n-1} q_r(1-p)^r+q_n(1-p)^n.
\]
Since $(q_r)_{r\ge 1}$ is nonincreasing and
\[
    (1-p)^n=p\sum_{r\ge n}(1-p)^r,
\]
this quantity is at least
\[
    p\sum_{r\ge 1}q_r(1-p)^r.
\]
A numerical maximization over $p\in[0,1]$ gives the value $0.4477\ldots$, attained at $p\approx 0.2029$.

For reproducibility, we compute the classical $r$-candidate value as
\[
q_r=
\max\left\{
\frac{1}{r},
\max_{a\in\{1,\dots,r-1\}}
\frac{1}{r}\sum_{i=a+1}^r \frac{a}{i-1}
\right\}.
\]
Here the first term corresponds to selecting the first candidate, and the inner maximum corresponds to rejecting the first $a\ge 1$ candidates. To compute the numerical value, we fix a cutoff $R$, evaluate the terms with $r<R$ exactly, and bound every remaining term using $q_r\le q_R$ for $r\ge R$. This is valid because $(q_r)$ is nonincreasing and $q_r\to 1/e$. We use the standard monotonicity of the finite secretary value $q_r$: when the number of candidates increases, the optimal worst-case success probability cannot increase.

\textbf{Optimality.} It remains to prove optimality. The point that requires care is that a one-query algorithm need not query after a fixed prefix, and it could try to select a record before using its query. We rule out both possibilities. This reduction is included because the upper bound must apply to arbitrary one-query ordinal algorithms, not only to the algorithm described above.

\begin{lemma}
\label{lem:ro1-post-query}
Fix a pre-query ordinal history $H$ ending after $s$ arrivals, and condition on the event $\mathrm{QR}_s=r$. After receiving this response, no algorithm can have conditional success probability larger than $q_r$.
\end{lemma}

\begin{proof}[Proof of Lemma~\ref{lem:ro1-post-query}]
Conditional on $H$ and on $\mathrm{QR}_s=r$, the $r$ unseen elements above the current threshold still arrive in a uniformly random relative order.

The only elements that can be the global maximum are the $r$ unseen elements whose weights are above $\tau(E_s)$. All other unseen elements have weight at most $\tau(E_s)$ and can never be a winning selection. Conditional on $\mathrm{QR}_s=r$, we can therefore ignore these lower-weight arrivals and view the post-query problem as the classical secretary problem on the $r$ relevant elements, which arrive in uniformly random relative order.

If a post-query strategy had conditional success probability larger than $q_r$, then it would yield a classical $r$-candidate secretary strategy with success probability larger than $q_r$, contradicting the definition of $q_r$.
\end{proof}

\begin{lemma}
\label{lem:ro1-no-selection-before-query}
For every one-query ordinal algorithm in the random-order model, there is another one-query ordinal algorithm with at least the same success probability that makes no selection before using its query.
\end{lemma}

\begin{proof}[Proof of Lemma~\ref{lem:ro1-no-selection-before-query}]
It is enough to consider deterministic algorithms, since a randomized algorithm is a distribution over deterministic ones. Also, before the query is used, a deterministic ordinal algorithm never benefits from selecting a non-record. Selecting a non-record cannot win, and rejecting it preserves all future options.

Fix a deterministic algorithm $A$. Suppose that, for some pre-query history $H$, the algorithm would select the next record before using its query. Choose such a history $H$ of minimum length. Let $s$ be the number of elements observed in $H$, and let $h=\mathrm{QR}_s$. Construct a modified algorithm $A_H$ that behaves as $A$ outside histories extending $H$, but after $H$ it rejects the current prefix, queries immediately, and runs an optimal classical secretary rule on the $h$ elements above the current threshold.

We compare the two algorithms conditional on $H$ and on $h$. This use of Lemma~\ref{lem:ro1-post-query} is conditional on the ordinal history $H$. It is valid because, under random order, once $H$ and $h=\mathrm{QR}_s$ are fixed, the relative order of the $h$ unseen elements above the current threshold is still uniform.

If no above-threshold element appears before the time at which $A$ would have used its query, then $A$ reaches its query with the same set of $h$ above-threshold elements still unseen. By Lemma~\ref{lem:ro1-post-query}, its conditional success probability from that point is at most $q_h$, which is exactly the success probability guaranteed by $A_H$ after querying at $H$.

If an above-threshold element appears before that query time, then the first such element is the next record. By the choice of $H$, algorithm $A$ selects it before using its query. Conditional on $H$ and $h$, the relative order of the $h$ above-threshold elements is uniform, so this selected element is the maximum among them with probability $1/h$. Since $q_h\ge 1/h$ for every $h\ge 1$, the modified algorithm $A_H$ is again at least as good. If $h=0$, both algorithms have already missed the global maximum.

Thus the modification does not decrease the success probability and removes one earliest pre-query selection node from the strategy tree. Repeating this finite modification yields an algorithm with at least the same success probability that never selects before querying.
\end{proof}

\begin{lemma}
\label{lem:ro1-fixed-prefix}
Among one-query ordinal algorithms that make no selection before querying, there is an optimal algorithm whose query time is deterministic. More precisely, for such an algorithm it is enough to reject a fixed number $s$ of elements, query, and then use an optimal post-query rule.
\end{lemma}

\begin{proof}[Proof of Lemma~\ref{lem:ro1-fixed-prefix}]
Let $A$ be a deterministic ordinal algorithm that makes no selection before querying, and let $T$ be its query time. The event $T=s$ is determined by the ordinal pattern of the first $s$ arrivals. In a uniformly random arrival order, conditional on any fixed ordinal pattern of the first $s$ arrivals, the set $E_s$ of elements seen in the first $s$ positions is still a uniformly random $s$-subset of $E$. Therefore the distribution of $\mathrm{QR}_s$ conditional on $T=s$ is the same as the distribution of $\mathrm{QR}_s$ for a fixed query time $s$.

By Lemma~\ref{lem:ro1-post-query}, after querying at time $s$ the conditional success probability is at most $q_{\mathrm{QR}_s}$. Hence the success probability of $A$ is a convex combination (over the possible values of $s$) of the success probabilities obtained by querying after a fixed prefix of length $s$ and then using the optimal post-query rule. This convex combination is no larger than the best fixed-prefix value. Thus there exists a deterministic query time $s$ that is at least as good as $A$.
\end{proof}

By Lemmas~\ref{lem:ro1-no-selection-before-query} and~\ref{lem:ro1-fixed-prefix}, it remains only to upper bound algorithms that reject a deterministic number $s=s_n$ of elements, query, and then use the optimal post-query rule.

If $s_n=0$, then $\mathrm{QR}_{s_n}=n$ and the success probability is $q_n$, which tends to $1/e$. If $s_n=n$, then the success probability is zero. Hence assume $1\le s_n\le n-1$.

For $1\le r\le n-s_n$,
\[
    \PP[\mathrm{QR}_{s_n}=r]
    =\frac{\binom{n-r-1}{s_n-1}}{\binom{n}{s_n}}
    =\frac{s_n}{n}\cdot \frac{(n-s_n)_{\underline{r}}}{(n-1)_{\underline{r}}}.
\]
Indeed, the event $\mathrm{QR}_{s_n}=r$ means that $x^1,\dots,x^r$ are not in the prefix and that $x^{r+1}$ is in the prefix.

Let $s_n/n\to p$ along a convergent subsequence. If $p\in(0,1)$, then for each fixed $r$,
\[
    \PP[\mathrm{QR}_{s_n}=r]\to p(1-p)^r.
\]
Moreover, the summands are dominated by a summable geometric sequence, so dominated convergence gives the limiting upper bound
\[
    p\sum_{r\ge 1} q_r(1-p)^r.
\]
If $p=1$, then the probability that the maximum lies after the query tends to zero, so the limiting success probability is zero.

If $p=0$, then for every fixed $K$ the probability that $\mathrm{QR}_{s_n}<K$ tends to zero, while the contribution from $\mathrm{QR}_{s_n}\ge K$ is at most $q_K$. Letting $K\to\infty$ gives a limiting upper bound of $1/e$. This is also at most the maximum of $p\sum_{r\ge 1}q_r(1-p)^r$, because Algorithm~\ref{alg:ro1} attains $0.4477\ldots>1/e$.

Therefore every one-query algorithm has asymptotic success probability at most
\[
    \max_{p\in[0,1]} p\sum_{r\ge 1}q_r(1-p)^r.
\]
This is exactly the value achieved by Algorithm~\ref{alg:ro1}. The numerical maximization gives $0.4477\ldots$ at $p\approx 0.2029$, completing the proof.
\end{proof}

\subsubsection{Unlimited Query Rounds}

Next, we consider the Random Order model with unlimited query rounds. Consider Algorithm~\ref{alg:rounl}, which queries the oracle after rejecting each element. It keeps rejecting while at least two unseen elements beat the current record. If the response becomes $1$, it selects the next record. If the response becomes $0$, it has already rejected the maximum and terminates without selecting.

For $n=1$, the algorithm simply selects the unique element. In the pseudocode below, we assume $n\ge2$.

\begin{algorithm}
    \DontPrintSemicolon
    $r\leftarrow n$\;
    \While{$r>1$}{
        reject the next arriving element, say at time $t$, and query to obtain $r\gets \mathrm{QR}_t$\;
    }
    \If{$r=0$}{terminate without selecting an element\;}
    let $E_t$ be the observed set at this time and set $\tau\gets\tau(E_t)$\;
    \Return the first subsequently arriving element $x$ such that $w(x)>\tau$\;
	\caption{First After Second Algorithm}
	\label{alg:rounl}
\end{algorithm}
\begin{theorem}
\label{theo:rounl}
    Algorithm~\ref{alg:rounl} is $\frac{1}{2}$-competitive and optimal in the Random Order with Unlimited Query Rounds model.
\end{theorem}

\begin{proof}
    The algorithm selects $x^1$ if and only if $x^2$ arrives earlier, which happens with probability $1/2$ under random order. For optimality, it suffices to consider the case $n=2$. The oracle provides no information before rejecting the first element, and after rejecting it the second element is the only remaining option.
\end{proof}

\section{Matroid Secretary Problems}

We now turn to feasibility constraints. The Matroid Secretary Problem generalizes the single-choice problem by imposing feasibility constraints: the candidates form the ground set of a matroid, and the decision-maker may accept multiple candidates as long as the accepted set remains independent throughout. Below, we recap basic definitions and properties of matroids.

\begin{definition}[Matroid]
    A matroid is a tuple $M = (E,\mathcal{I})$, where $E$ is a finite set and $\mathcal{I}\subseteq 2^E$ is a collection that satisfies the following axioms:
    \begin{enumerate}
        \item If $I\subseteq J$ and $J \in \mathcal{I}$, then $I \in \mathcal{I}$.
        \item If $I,J\in\mathcal{I}$ and $|I|>|J|$, then there exists $x\in I\setminus J$ such that $J+x\in \mathcal{I}$.
    \end{enumerate}
\end{definition}

The set $E$ is called the \emph{ground set}, and the elements of $\mathcal{I}$ are called \emph{independent sets}. An inclusion-wise maximal independent set is called a \emph{base} of $M$. A basic property of matroids is that all bases have the same cardinality.

\begin{definition}[Rank]
    The rank of a matroid $M = (E,\mathcal{I})$, denoted $\mathrm{rank}(M)$, is the cardinality of a base of $M$. For $I \subseteq E$, we define $\mathrm{rank}_M(I)$ as the cardinality of a maximal independent subset of $I$ (we also write $\mathrm{rank}(I)$ when the matroid is clear from context).
\end{definition}

\begin{definition}[Span]
    Let $M = (E,\mathcal{I})$ be a matroid. For $I \subseteq E$, the span of $I$ is defined as
    \begin{align*}
        \spa(I) = \{x\in E : \mathrm{rank}(I+x)=\mathrm{rank}(I)\}.
    \end{align*}
\end{definition}

\begin{algorithm}

	\KwIn{Matroid $M = (E,\mathcal{I})$, $\sigma:[|E|]\rightarrow E$ order of $E$}
    \KwOut{Base $\ALG$ of $M$}
    $\ALG \gets \emptyset$\;
    \For{$i\in[|E|]$}{
        $x \gets \sigma(i)$\;
        \lIf{$\ALG + x \in \mathcal{I}$}{$\ALG \gets \ALG + x$}
    }
    \Return $\ALG$
	\caption{Greedy Algorithm}
	\label{alg:greedy}
\end{algorithm}
The greedy algorithm (Algorithm~\ref{alg:greedy}) is a useful tool for studying optimization problems with matroid constraints. We denote its output by $\GREEDY(M,\sigma)$. We will use the following two basic properties.

\begin{proposition}
\label{prop:greedy_any}
    Let $M = (E,\mathcal{I})$ be a matroid, let $\sigma$ be an order of $E$, and let $x \in E$. Then $x \in \GREEDY(M,\sigma)$ if and only if $x \notin \spa(\{z\in E : \sigma^{-1}(z)<\sigma^{-1}(x)\})$ (i.e., $x$ is not spanned by the elements that appear before $x$).
\end{proposition}

\begin{proof}
Let $G=\GREEDY(M,\sigma)$. For each prefix, let
\[
    E_i:=\{z\in E:\sigma^{-1}(z)\le i\},
    \qquad
    G_i:=G\cap E_i.
\]
We first show that $G_i$ is a maximal independent subset of $E_i$. Suppose not. Then there exists $e\in E_i\setminus G_i$ such that $G_i+e$ is independent. Let $t=\sigma^{-1}(e)\le i$. When $e$ was processed, the greedy set was $G_{t-1}\subseteq G_i$. Since $G_i+e$ is independent, $G_{t-1}+e$ is also independent. Thus greedy would have added $e$, a contradiction.

Therefore $G_i$ is a base of the restriction $M|E_i$. Hence $G_i$ has rank $\mathrm{rank}(E_i)$ and spans $E_i$, so
\[
    \spa(G_i)=\spa(E_i)
\]
for every $i$.

Now let $j=\sigma^{-1}(x)$. When $x$ is processed, the greedy algorithm has selected $G_{j-1}$. It accepts $x$ if and only if $G_{j-1}+x$ is independent, which is equivalent to $x\notin\spa(G_{j-1})$. Since $\spa(G_{j-1})=\spa(E_{j-1})$, this is equivalent to
\[
    x\notin \spa(\{z\in E:\sigma^{-1}(z)<\sigma^{-1}(x)\}).
\]
\end{proof}

\begin{proposition}
\label{prop:greedy_weight}
    Let $M = (E,\mathcal{I})$ be a matroid, let $w:E\rightarrow \RR_{>0}$ be an injective weight function, and let $\sigma_w$ be the order of $E$ in decreasing order of weight. Then the unique maximum-weight base of $M$ is $\OPT = \GREEDY(M,\sigma_w)$. Furthermore, for $x\in E$, we have $x\in\OPT$ if and only if $x\notin\spa(\{z\in E : w(z)>w(x)\})$ (by Proposition~\ref{prop:greedy_any}).
\end{proposition}

\begin{proof}
    The usual greedy proof for matroids shows that the greedy algorithm in decreasing weight order returns a maximum-weight base. We recall the uniqueness argument. Suppose that $B$ and $B'$ are two distinct maximum-weight bases. Let $e$ be the maximum-weight element in $B\triangle B'$, and assume $e\in B$. By basis exchange, there is $f\in B'\setminus B$ such that $B'-f+e$ is a base. Since the weights are injective and $e$ is the maximum-weight element of the symmetric difference, we have $w(e)>w(f)$. This contradicts the optimality of $B'$. The span characterization follows from Proposition~\ref{prop:greedy_any} applied to the greedy order.
\end{proof}

\subsection{Ordinal queries in matroid secretary problems}

Let $M = (E,\mathcal{I})$ be a matroid. In the Matroid Secretary Problem, the elements of $E$ arrive and reveal their weights according to the same online rules as in the single-choice problem. The algorithm may select multiple elements, as long as the selected set remains independent in $M$.

We assume that the matroid $M$ is known to the algorithm in advance and that the algorithm has access to an independence oracle, which allows it to test whether a subset $I\subseteq E$ is independent. We consider the same arrival-order models as before, but the oracle returns a different ordinal response. In particular, we define two types of query responses: simple query responses and complete query responses. Informally, a simple query response identifies which unseen elements would improve the current optimum, while a complete query response further refines this information by partitioning unseen elements into ordinal weight buckets relative to $\OPT(E_i)$.

In the matroid models, a query round may be made after the algorithm has made its irrevocable accept/reject decision for the current element and before the next element arrives. Thus, unlike in the single-choice setting, the process may continue after an accepted element. The observed set $E_i$ contains all elements that have arrived by time $i$, regardless of whether they were accepted or rejected.

For $F\subseteq E$, let $M|F$ be the restriction of $M$ to $F$. Since all weights are positive, every maximum-weight independent set in $M|F$ is a base of $M|F$. We write $\OPT(F)$ for the unique maximum-weight base of $M|F$.

Loops can be ignored, since they are never contained in an independent set. We therefore assume in the algorithms below that the matroid has no loops. If the matroid has rank zero, then $\OPT(E)=\emptyset$. The algorithm returns the empty set, and the probability-competitive guarantee is vacuous. We assume positive rank in the rest of the argument.

\begin{definition}[Simple query response]
\label{def:simple}
    For a given arrival order $\pi$, weight function $w$, and $i\in\{0,1,\dots,n\}$, the simple query response at step $i$, denoted $\mathrm{SQR}_i$, is defined as
    \begin{align*}
        \mathrm{SQR}_i := \{x\in E\setminus E_i : x \in \OPT(E_i + x)\},
    \end{align*}
    that is, the set of elements that arrive after step $i$ and improve the current optimal solution at step $i$.
\end{definition}

\begin{definition}[Complete query response]
    For a given arrival order $\pi$, weight function $w$, and $i\in\{0,1,\dots,n\}$, let $\OPT(E_i) = \{z^1,\dots,z^m\}$ be the current optimal solution at step $i$, where $w(z^1)>w(z^2)>\cdots>w(z^m)$. The case $m=0$ is allowed. The complete query response at step $i$, denoted $\mathrm{CQR}_i$, is defined as the ordered tuple $(A_0,A_1,\dots,A_m)$, where
    \begin{align*}
        A_j := \{x\in E\setminus E_i : w(z^{j+1})<w(x)<w(z^j)\}
    \end{align*}
    for each $j \in \{0,\dots,m\}$, where $w(z^0):=\infty$ and $w(z^{m+1}):=0$.
    Because all weights are positive and injective, the sets $A_0,\ldots,A_m$ form a partition of $E\setminus E_i$.
    In particular, if $i=0$, then $\OPT(E_0)=\emptyset$ and
    $\mathrm{CQR}_0=(E)$.
\end{definition}

Both matroid responses reveal identities of unseen elements. This is intentional. The restriction is not on the size of the response, but on the kind of information certified by the oracle: membership in $\mathrm{SQR}_i$ and in the buckets of $\mathrm{CQR}_i$ is determined only by ordinal comparisons with the current optimum on $E_i$ and by the known matroid structure. Thus the oracle may identify which unseen elements satisfy an ordinal condition, but it never reveals their numerical weights.

The algorithms we consider for the Matroid Secretary Problem with One Simple Query and with One Complete Query (Algorithms~\ref{alg:msps} and~\ref{alg:mspc1}, respectively) are divided into two phases: a sampling phase and a preselection phase. In the sampling phase, a random sample arrives; its elements' weights are observed, but no element is selected. After the sampling phase, the algorithm makes one oracle call and receives the simple query response or complete query response, depending on the model.

Next, in the preselection phase, the algorithm chooses its output using only the information observed so far and the query response received, without learning any additional information about future weights. Equivalently, one can think of the algorithm as pausing the process after receiving the query response, deciding which elements to select before observing their weights, and then resuming the process and selecting exactly those elements. This is only a description of the computation performed after the query. Online, the algorithm stores the identities it has preselected and accepts exactly those elements when they arrive.

We say that an algorithm $A$ for a model $\mathcal{M}$ of the matroid secretary problem is $\alpha$-probability-competitive if, for every element $e$ in the maximum-weight base $\OPT$, the probability that $A$ selects $e$ is at least $\alpha$. This notion is stronger than the usual utility-competitive ratio: if every element of $\OPT$ is selected with probability at least $\alpha$, then the expected weight of the output is at least $\alpha$ times the weight of $\OPT$ (see, e.g.,~\cite{SotoTV2021}).

All matroid algorithms below are polynomial-time given an independence oracle for $M$, apart from the cost of receiving the query response itself. The set $\OPT(F)$ is computed by the standard greedy algorithm after sorting the observed elements by weight. Span tests can be implemented through rank computations, and each rank computation can be performed by greedy using polynomially many independence-oracle calls. The query responses may have size $\Theta(n)$; throughout the paper we count query rounds rather than the bit-length of the response.

\subsection{Sample-Based Online Adversarial Order with One Simple Query}

\begin{algorithm}

    \KwIn{Matroid $M = (E,\mathcal{I})$}
    \KwOut{Independent set $\ALG$ of $M$}
    $\ALG \gets \emptyset$, $s \gets \Bin(n,\frac{1}{2})$\;
    \KwSty{Sampling phase:} reject a uniform random sample $S\subseteq E$ of size $s$\;
    $Z = \{z^1,...,z^m\} \gets \OPT(S)$, where $w(z^1)>w(z^2)>...>w(z^m)$\;
    \For{$i\in[m]$}{
        \For{$x\in \left(\spa(\{z^1,...,z^i\})\setminus\spa(\{z^1,...,z^{i-1}\})\right)\setminus S$ in any order}{
            \lIf{$\ALG + x \in \mathcal{I}$ and $w(x)>w(z^i)$}{$\ALG \gets \ALG + x$}
        }
    }
    \For{$x\in E\setminus\spa(Z)$ in any order}{
        \lIf{$\ALG + x \in \mathcal{I}$}{$\ALG \gets \ALG + x$}
    }
    \Return $\ALG$
	\caption{(Jaillet, Soto, and Zenklusen) Algorithm for MSP Free Order Model}
	\label{alg:free}
\end{algorithm}
In \cite{JailletSZ13} the authors study the \emph{Free Order} model for the Matroid Secretary Problem and propose Algorithm~\ref{alg:free}. Informally, Free Order gives the algorithm post-sample control over the \emph{revelation order} of the unseen elements: it may adaptively choose which unseen element is shown next.

In the full version of their paper, Jaillet, Soto, and Zenklusen~\cite{JailletSZ13} show the following.

\begin{theorem}[\cite{JailletSZ13}]
\label{theo:free}
    Algorithm~\ref{alg:free} is $\frac{1}{4}$-probability-competitive for the Matroid Secretary Problem in the Free Order model.
\end{theorem}

Algorithm~\ref{alg:msps} achieves the same guarantee in SBOnAO with one simple query by replacing control over the revelation order with one round of ordinal queries. After sampling, it computes $Z=\OPT(S)$ and queries once for $\mathrm{SQR}_s$ (Definition~\ref{def:simple}).

The simple-query algorithm transfers the free-order algorithm of Jaillet, Soto, and Zenklusen to the sample-based adversarial model. In the free-order model, the algorithm uses its ability to choose the next revealed element to determine which unseen elements fall into the improving classes induced by $Z=\OPT(S)$. In our model the post-sample order is adversarial, so this identification step cannot be carried out by controlling arrivals. The simple query provides exactly the missing ordinal information by returning $\mathrm{SQR}_s$. Given $Z$ and $\mathrm{SQR}_s$, Algorithm~\ref{alg:msps} precomputes the same greedy outcome and then simply accepts precisely those precomputed elements when they happen to arrive.

We make this coupling formal in the proof of Theorem~\ref{theo:msps1}.

\begin{algorithm}

    \KwIn{Matroid $M = (E,\mathcal{I})$}
    \KwOut{Independent set $\ALG$ of $M$}
    $\ALG \gets \emptyset$, $s \gets \Bin(n,\frac{1}{2})$\;
    \KwSty{Sampling phase:} reject a uniform random sample $S\subseteq E$ of size $s$\;
    $Z = \{z^1,...,z^m\} \gets \OPT(S)$, where $w(z^1)>w(z^2)>...>w(z^m)$\;
    $G \gets \mathrm{SQR}_s$\;
    \textbf{Preselection phase:} without seeing the remaining weights:\\
    \For{$i\in[m]$}{
        $G_i \gets G\cap\spa(\{z^1,...,z^i\})\setminus\spa(\{z^1,...,z^{i-1}\})$\;
        \For{$x\in G_i$ in any order}{
            \lIf{$\ALG + x \in \mathcal{I}$}{$\ALG \gets \ALG + x$}
        }
    }
    \For{$x\in E\setminus\spa(Z)$ in any order}{
        \lIf{$\ALG + x \in \mathcal{I}$}{$\ALG \gets \ALG + x$}
    }
    \Return $\ALG$
	\caption{Post-Sample Selection with Simple Query Algorithm}
	\label{alg:msps}
\end{algorithm}
\begin{theorem}
\label{theo:msps1}
    Algorithm~\ref{alg:msps} is $\frac{1}{4}$-probability-competitive for the Matroid Secretary Problem in the Sample-Based Online Adversarial Order model with One Simple Query.
\end{theorem}

\begin{proof}
Fix the sample $S$ and let $Z=\OPT(S)=\{z^1,\dots,z^m\}$, ordered by decreasing weight. Couple Algorithms~\ref{alg:free} and~\ref{alg:msps} by using the same internal order in every loop labeled ``in any order''. We show that the two algorithms perform the same independence tests, in the same order.

For each $i\in[m]$, define
\[
    L_i :=
    \bigl(\spa(\{z^1,\dots,z^i\})
    \setminus
    \spa(\{z^1,\dots,z^{i-1}\})\bigr)\setminus S.
\]
Algorithm~\ref{alg:free} tests the elements $x\in L_i$ with $w(x)>w(z^i)$. Algorithm~\ref{alg:msps} tests the elements $x\in L_i\cap G$, where $G=\mathrm{SQR}_s$. We prove that, for every $i\in[m]$ and every $x\in L_i$,
\[
    x\in G
    \quad\Longleftrightarrow\quad
    w(x)>w(z^i).
\]

By Definition~\ref{def:simple}, $x\in G$ if and only if $x\in\OPT(S+x)$. We use the greedy characterization in Proposition~\ref{prop:greedy_weight}. During the greedy computation of $\OPT(S)$, after the algorithm has processed all sample elements of weight larger than a value $\lambda$, the selected elements span all sample elements of weight larger than $\lambda$.

First suppose that $w(x)>w(z^i)$. Let $h$ be the number of elements of $Z$ whose weight is larger than $w(x)$. Then $h<i$, so $\{z^1,\dots,z^h\}\subseteq\{z^1,\dots,z^{i-1}\}$. The sample elements with weight larger than $w(x)$ are spanned by $\{z^1,\dots,z^h\}$. Since $x\notin\spa(\{z^1,\dots,z^{i-1}\})$, we also have $x\notin\spa(\{z^1,\dots,z^h\})$. Hence $x$ is not spanned by the elements of $S$ with weight larger than $w(x)$. By Proposition~\ref{prop:greedy_weight}, $x\in\OPT(S+x)$, and therefore $x\in G$.

Conversely, suppose that $w(x)<w(z^i)$. Since $x\in L_i$, we have $x\in\spa(\{z^1,\dots,z^i\})$. All elements $z^1,\dots,z^i$ have weight larger than $w(x)$, so $x$ is spanned by elements of $S$ with weight larger than $w(x)$. Proposition~\ref{prop:greedy_weight} implies that $x\notin\OPT(S+x)$, and therefore $x\notin G$.

This proves the equivalence. Therefore, in each layer $L_i$, Algorithm~\ref{alg:msps} identifies exactly the elements that Algorithm~\ref{alg:free} would test using free-order access. Moreover, since $Z$ is a base of the restricted matroid $M|S$, we have $S\subseteq\spa(Z)$. Hence $S\cap(E\setminus\spa(Z))=\emptyset$.

On the remaining set $E\setminus\spa(Z)$, both algorithms test the same elements independently of the query response. With the internal orders coupled as above, both algorithms perform the same independence tests in the same order and return the same set.

Algorithm~\ref{alg:msps} precomputes an independent set, and the elements selected online always form a subset of it. Thus the online decisions maintain independence.

The $1/4$ probability-competitive guarantee follows from Theorem~\ref{theo:free}.
\end{proof}
\subsection{Sample-Based Online Adversarial Order with One Complete Query}

Next, we study the analogous model in which the oracle provides a stronger query response. In the Complete Query models, the query response partitions the set of unseen elements into buckets, where each bucket corresponds to elements whose weights lie between two consecutive weights in $\OPT(S)$. We propose Algorithm~\ref{alg:mspc1} for the Sample-Based Online Adversarial Order model with One Complete Query.

The random orders used inside the buckets are internal random choices of the algorithm for the precomputed greedy order. They are not assumptions on the post-sample arrival order.

The complete query is stronger than the simple query because it gives the ordinal bucket of each unseen element relative to the current sample optimum. The algorithms below use this information only to define a greedy order before the selection phase resumes. They do not use numerical weights of unseen elements.

As in Algorithm~\ref{alg:msps}, Algorithm~\ref{alg:mspc1} fixes an order $\sigma$ of $(E\setminus S)\cup Z$ after receiving the query response so that the elements of $A_0$ appear first, then $z^1$, then the elements of $A_1$, then $z^2$, and so on until $A_m$. The algorithm then computes $\AUX = \GREEDY(M',\sigma)$, where $M'$ is the restriction of $M$ to $\bigcup_{i=0}^m A_i \cup Z$. Since the elements of $Z$ were previously rejected (they appeared in the sample), the algorithm outputs $\AUX \setminus Z$.

Complete queries give a finer ordinal partition. The proof below uses this partition to isolate, for each fixed $y\in\OPT$, the elements that can prevent $y$ from being selected. This is the only point where the argument becomes more technical than the single-choice proof: in a matroid, an element can be blocked by dependence, not only by a heavier element arriving first. The blocking set $B$ records exactly the lower-weight elements in the same complete-query bucket that are not already spanned by the heavier part of the sample optimum.

\begin{algorithm}

    \KwIn{Matroid $M = (E,\mathcal{I})$}
    \KwOut{Independent set $\ALG$ of $M$}
	\KwSty{Parameters:} $p\in[0,1]$\;
    $\AUX \gets \emptyset$, $s \gets \Bin(n,p)$\;
    \KwSty{Sampling phase:} reject a uniform random sample $S\subseteq E$ of size $s$\;
    $Z = \{z^1,...,z^m\} \gets \OPT(S)$, where $w(z^1)>w(z^2)>...>w(z^m)$\;
    $(A_0,A_1,...,A_m) \gets \mathrm{CQR}_s$\;
    \textbf{Preselection phase:} without seeing the remaining weights:\\
    \For{$x\in A_0$ in random order}{
        \lIf{$\AUX + x \in \mathcal{I}$}{$\AUX \gets \AUX + x$}
    }
    \For{$i\in[m]$}{
        \lIf{$\AUX + z^i \in \mathcal{I}$}{$\AUX \gets \AUX + z^i$}
        \For{$x\in A_i$ in random order}{
            \lIf{$\AUX + x \in \mathcal{I}$}{$\AUX \gets \AUX + x$}
        }
    }
    $\ALG \gets \AUX\setminus Z$\;
    \Return $\ALG$
	\caption{Post-Sample Selection with Complete Query Algorithm}
	\label{alg:mspc1}
\end{algorithm}
\begin{theorem}
\label{theo:mspc1}
    Setting $p = \frac{1}{e}$, Algorithm~\ref{alg:mspc1} is $\frac{1}{e}$-probability-competitive and optimal for the Matroid Secretary Problem in the Sample-Based Online Adversarial Order model with One Complete Query.
\end{theorem}

\begin{proof}
    \textbf{Lower bound.} Let $y \in \OPT$. We show that
    $\PP[y\in\ALG]\geq 1/e$. Let $Z^0:=\emptyset$, and for
    $i\in[m]$ let $Z^i:=\{z^1,\dots,z^i\}$.

    For a fixed sample $S$, let $j$ be the number of elements of
    $Z=\OPT(S)$ whose weight is larger than $w(y)$. Thus $Z^j$ is the prefix
    of $Z$ formed by the elements heavier than $y$. For $a\in\mathbb R$, write
    $E_{\le a}:=\{x\in E:w(x)\le a\}$ and
    $E_{>a}:=\{x\in E:w(x)>a\}$.

    If $y\in S$, set $B:=\emptyset$. If $y\notin S$, then $y$ belongs to the
    bucket $A_j$ of $\mathrm{CQR}_s$, and we define
    \[
        B := A_j\cap E_{\leq w(y)}\setminus \spa(Z^j).
    \]
    Thus $B$ is the set of elements in the same bucket as $y$, with weight at
    most $w(y)$, that are not already spanned by the heavier prefix $Z^j$.
    The elements that can ``block'' $y$ from being selected are contained in
    $B$.

    More precisely, suppose that $y\notin S$ and that $y$ is the first element
    of $B$ to appear in the uniformly random internal order of $A_j$. Then the
    algorithm selects $y$. Indeed, after the elements of $Z^j$ are processed,
    the current set $\AUX$ spans $Z^j$: each $z^\ell$ with $\ell\le j$ is either
    added, or is already spanned by the current $\AUX$ at its turn. Hence no
    element of $\spa(Z^j)\cap A_j$ is added later. Therefore every element added
    to $\AUX$ before $y$ is heavier than $y$. Since $y\in\OPT$,
    Proposition~\ref{prop:greedy_weight} implies
    $y\notin\spa(E_{>w(y)})$. Thus $y$ is added to $\AUX$ and then to $\ALG$.

    Since the order of elements in $B$ is uniformly random,
    \begin{align*}
        \PP[y\in\ALG\mid |B|=t]\geq
        \begin{cases}
            \frac{1}{t}, & \text{if $t>0$,}\\
            0, & \text{if $t=0$.}
        \end{cases}
    \end{align*}

    The last step in the proof is to determine the distribution of $|B|$. Let $E = \{x^1,\dots,x^n\}$ in decreasing order of weight, and view the sampling step as tossing independent coins $C^1,\dots,C^n$, where $C^i=\mathrm{heads}$ with probability $p$ and $C^i=\mathrm{tails}$ with probability $(1-p)$, and $x^i\in S$ iff $C^i=\mathrm{heads}$. Let $k\in[n]$ be such that $x^k=y$.
    We claim the following.

    \begin{claim}
        Fix any realization of $C^1,\dots,C^{k-1}$. Under this conditioning, $j$ and $Z^j$ are determined. Let $H := E_{\leq w(y)}\setminus \spa(Z^j)$ and $h:=|H|$. Since $y\in\OPT(E)$ and all elements of $Z^j$ have weight larger than $w(y)$, Proposition~\ref{prop:greedy_weight} implies $y\notin\spa(Z^j)$; hence $y\in H$ and $h\ge 1$. Then
        \[
            \PP\bigl[|B|=t\mid C^1,\dots,C^{k-1}\bigr]=
            \begin{cases}
                p(1-p)^t, & t=0,1,\dots,h-1,\\
                (1-p)^h, & t=h.
            \end{cases}
        \]
    \end{claim}

    We now prove the claim. Under a fixed realization of $C^1,\dots,C^{k-1}$, the index $j$ (and thus $Z^j$ and $H$) is fixed. If $C^k=\mathrm{heads}$ (i.e., $y\in S$), then $B=\emptyset$ by definition, so $|B|=0$. Hence $\PP\bigl[|B|=0\mid C^1,\dots,C^{k-1}\bigr]=p$.

    Suppose now that $C^k=\mathrm{tails}$ (i.e., $y\notin S$). List the elements of $H$ in decreasing order of weight as $u^1,\dots,u^h$, and write $u^\ell=x^{i_\ell}$ for each $\ell\in[h]$. If some element of $H$ lies in $S$, then $z^{j+1}$ is the heaviest such element. Equivalently, its coin is the first $\mathrm{heads}$ among $C^{i_1},\dots,C^{i_h}$. If no element of $H$ lies in $S$, then there is no such $z^{j+1}$.

    Consequently, $|B|$ is the number of consecutive $\mathrm{tails}$ before this first $\mathrm{heads}$, truncated at $h$ if no $\mathrm{heads}$ occurs. Therefore, for $t\in\{1,\dots,h-1\}$ we have $\PP\bigl[|B|=t\mid C^1,\dots,C^{k-1},\; C^k=\mathrm{tails}\bigr]=p(1-p)^{t-1}$, and $\PP\bigl[|B|=h\mid C^1,\dots,C^{k-1},\; C^k=\mathrm{tails}\bigr]=(1-p)^{h-1}$.
    Multiplying by $\PP[C^k=\mathrm{tails}]=1-p$ gives the stated conditional probabilities, completing the proof of the claim.

    Combining the expressions, we obtain

    \begin{align*}
        \PP[y\in\ALG] &= \sum_{c^1,...,c^{k-1}}\PP[y\in\ALG| C^1=c^1,...,C^{k-1}=c^{k-1}]\cdot\PP[C^1=c^1,...,C^{k-1}=c^{k-1}]\\
        &\geq \sum_{c^1,...,c^{k-1}}\left( \sum_{t=1}^{h-1} \frac{p(1-p)^t}{t} + \frac{(1-p)^h}{h}\right)\cdot\PP[C^1=c^1,...,C^{k-1}=c^{k-1}] \\
        &\geq \sum_{c^1,...,c^{k-1}}\left( \sum_{t=1}^{\infty} \frac{p(1-p)^t}{t}\right)\cdot\PP[C^1=c^1,...,C^{k-1}=c^{k-1}] \\
        &= \sum_{t=1}^{\infty} \frac{p(1-p)^t}{t} = -p\ln(p),
    \end{align*}
where the second inequality holds since
$\sum_{t=h}^{\infty}\frac{p(1-p)^t}{t}\le \frac{(1-p)^h}{h}$.

Setting $p = 1/e$ implies that the algorithm is $1/e$-probability-competitive.

    \textbf{Optimality.} For optimality, it suffices to consider the rank-one uniform matroid. In this case, the matroid secretary problem reduces to the single-choice secretary problem. The only additional point is that a complete query may reveal the identity of all elements in
\[
    B:=\{x\in E\setminus S:w(x)>\tau(S)\},
\]
rather than only its size. We now show that this identity information does not improve the worst-case guarantee.

Fix a deterministic ordinal algorithm and consider the rank-one uniform matroid. Following Yao's principle, we first assign a strict weight order to the elements via a uniformly random permutation. Once the sample $S$ is realized, the complete query identifies the set of above-threshold unseen elements $B$.

Let $r=|B|$ and denote the decreasing weight order within $B$ by $y^1, y^2, \ldots, y^r$. The online adversary constructs the arrival order by choosing $i \in [r]$ uniformly at random, placing $y^1$ in the $i$-th position among the elements of $B$, and placing $y^r, y^{r-1}, \ldots, y^{r-i+2}$ before it. The remaining elements of $B$ follow $y^1$ in a fixed order, and elements outside $B$ arrive arbitrarily.

Conditional on $B$, the hidden weight order $y^1, \dots, y^r$ remains uniform over all $r!$ permutations of $B$. Thus, the element identities provide no additional ordinal information. By the same indistinguishability argument as in Theorem~\ref{theo:ao1}, applied within $B$, the algorithm selects $y^1$ for at most one choice of $i$. Hence, conditional on $S$ and $B$, the success probability is at most $1/r$ (and $0$ if $r=0$). By Yao's principle, the same conditional bound applies to randomized algorithms. Averaging over the query response $B$ (equivalently, over $r=|B|$) yields the same upper bound as in Theorem~\ref{theo:on1}, namely $1/e$. Hence Algorithm~\ref{alg:mspc1} is optimal.\qedhere
\end{proof}

Notice the similarities between this proof and the one for the single-choice case. Indeed, Algorithm~\ref{alg:on1} for the Single-Choice Sample-Based Online Adversarial Order model is a special case of Algorithm~\ref{alg:mspc1}, where the set $B$ becomes the set of unseen elements whose weights exceed the sample threshold.

\subsection{Random Order with Unlimited Complete Queries}

Next, we consider the case where unlimited calls to the oracle are allowed. We propose Algorithm~\ref{alg:mspcunl} for the Random Order model with Unlimited Complete Queries. The algorithm calls the oracle after each selection decision. It selects an element $x$ if $x$ is not spanned by the union of: (i) higher-weight elements already seen, (ii) buckets of higher weights, and (iii) the remainder of the bucket containing $x$. The initialization $A_0=E$ is the trivial bucket obtained before any element has arrived. It is not an additional query. The algorithm makes one complete query after each selection decision, and therefore uses at most $n$ query rounds.

\begin{algorithm}[H]
    \KwIn{Loopless matroid $M = (E,\mathcal{I})$}
    \KwOut{Independent set $\ALG$ of $M$}
    $\ALG \gets \emptyset$\;
    $A_0\gets E$, $m\gets 0$\;
    \tcp{Invariant: before processing arrival $i$, $(A_0,\ldots,A_m)=\mathrm{CQR}_{i-1}$; initially $A_0=E$.}
    \For{$i\gets 1$ \KwTo $n$}{
        observe arriving element $e_i$\;
        let $j$ be such that $e_i\in A_j$\;
        \uIf{$e_i\not\in\spa\left((\bigcup_{k=0}^j A_k\cup \{x\in E_{i-1}:w(x)>w(e_i)\})-e_i\right)$}{
            $\ALG \gets \ALG + e_i$\;
        }
        \If{$i<n$}{update $(A_0,\ldots,A_m)\gets\mathrm{CQR}_i$\;}
    }
    \Return $\ALG$
	\caption{MSP First After Second Algorithm}
	\label{alg:mspcunl}
\end{algorithm}
\begin{theorem}
\label{theo:mspcunl}

    Algorithm~\ref{alg:mspcunl} is $\frac{1}{2}$-probability-competitive and optimal for the Matroid Secretary Problem in the Random Order model with Unlimited Complete Queries.

\end{theorem}

\begin{proof}
  \textbf{Lower bound.} We prove the $\tfrac{1}{2}$ probability-competitive guarantee.

  Claim 1. Every element selected by Algorithm~\ref{alg:mspcunl} belongs to $\OPT(E)$.
  We first show that $\ALG\subseteq\OPT$ (in particular, the algorithm returns an independent set). Let $x\in\ALG$ be an element selected at time $t$. Let $(A_0,\dots,A_m)=\mathrm{CQR}_{t-1}$ be the tuple returned by the oracle after the previous step, and let $j\in\{0,\dots,m\}$ be such that $x\in A_j$. By the selection rule,
\[
    x\notin\spa\Bigl(\bigl(\bigcup_{k=0}^j A_k\cup \{z\in E_{t-1}: w(z)>w(x)\}\bigr)-x\Bigr).
\]
Moreover,
\[
\{z\in E:w(z)>w(x)\}
\subseteq
\Bigl(\bigcup_{k=0}^j A_k\cup \{z\in E_{t-1}:w(z)>w(x)\}\Bigr)-x.
\]
By monotonicity of span, the selection rule implies
\[
    x\notin \spa(\{z\in E:w(z)>w(x)\}).
\]
Hence, by Proposition~\ref{prop:greedy_weight}, we have $x\in\OPT$.

Since $\OPT(E)$ is independent and every selected element belongs to it, the final set returned by the algorithm is independent.

  Claim 2. For every $y\in \OPT(E)$, Algorithm~\ref{alg:mspcunl} selects $y$ with probability at least $1/2$.
  Now fix $y\in\OPT$. We show that $\PP[y\in\ALG]\ge \tfrac{1}{2}$. Write $E=\{x^1,\dots,x^n\}$ in decreasing order of weight. If $y\notin\spa(E-y)$, then $y$ is always added by the algorithm. Otherwise, there exists an index $i\in[n]$ such that
  \[
      y\in \spa(\{x^1,\dots,x^i\}-y)\setminus\spa(\{x^1,\dots,x^{i-1}\}-y),
  \]
  with the convention $\{x^1,\dots,x^{i-1}\}=\emptyset$ if $i=1$. Since $y\in\OPT$, we have $w(y)>w(x^i)$.

  If $x^i$ appears before $y$, then $x^i\in \OPT(E_{t-1})$. Indeed, suppose $x^i$ arrives at time $s$ and $y$ arrives at time $t$, with $s<t$. The set $E_{t-1}$ contains $x^i$ but not $y$. If $x^i\notin\OPT(E_{t-1})$, then by Proposition~\ref{prop:greedy_weight}, $x^i$ is spanned by the elements of $E_{t-1}$ with weight larger than $w(x^i)$. These elements are contained in $\{x^1,\dots,x^{i-1}\}\setminus\{y\}$, since $y$ has not arrived before time $t$. This would imply
\[
    x^i\in \spa(\{x^1,\dots,x^{i-1}\}-y),
\]
and hence
\[
    \spa(\{x^1,\dots,x^i\}-y)=\spa(\{x^1,\dots,x^{i-1}\}-y),
\]
contradicting the definition of $i$.

Since $x^i\in\OPT(E_{t-1})$, it appears as one of the boundary elements used by the complete query response $\mathrm{CQR}_{t-1}$. Therefore no bucket of $\mathrm{CQR}_{t-1}$ can contain elements on both sides of $x^i$ in the weight order. Let $j$ be such that $y\in A_j$ at time $t$. Since $w(y)>w(x^i)$, every element in $A_0\cup\cdots\cup A_j$ has weight larger than $w(x^i)$. Hence
\[
    \bigl(\bigl(\bigcup_{k=0}^j A_k\cup \{z\in E_{t-1}: w(z)>w(y)\}\bigr)-y\bigr)\subseteq \{x^1,\dots,x^{i-1}\}.
\]
Therefore $y$ is selected. Since $x^i$ appears before $y$ in a uniform random order with probability $\tfrac{1}{2}$, we conclude that the competitive ratio is at least $\tfrac{1}{2}$.

\textbf{Optimality.} For optimality, it suffices to consider the rank-one uniform matroid on two elements. Before the first accept/reject decision, the algorithm cannot have received any query response that depends on the first element. Hence it must either accept the first element, reject it, or randomize between these two actions. Since the arrival order is uniform, each element is the maximum with probability $1/2$. Therefore no algorithm, even with complete queries after the first decision, can guarantee success probability larger than $1/2$.\qedhere

\end{proof}

\paragraph*{Acknowledgements} This work was partially supported by ANID (Agencia Nacional de Investigaci\'on y Desarrollo, Chile) through FONDECYT Grant No.~1231669 and the BASAL Center for Mathematical Modeling (FB210005).
\bibliographystyle{plain}
\bibliography{references}

\begin{thebibliography}{10}

\bibitem{AntoniadisGKK23}
Antonios Antoniadis, Themis Gouleakis, Pieter Kleer, and Pavel Kolev.
\newblock {Secretary and online matching problems with machine learned advice}.
\newblock {\em Discrete Optimization}, 48(Part 2):100778, 2023.

\bibitem{azar2014prophet}
Pablo~Daniel Azar, Robert Kleinberg, and Seth~Matthew Weinberg.
\newblock {Prophet inequalities with limited information}.
\newblock In {\em Proceedings of {SODA} 2014}, pages 1358--1377. SIAM, 2014.

\bibitem{BabaioffIKK2018}
Moshe Babaioff, Nicole Immorlica, David Kempe, and Robert Kleinberg.
\newblock {Matroid secretary problems}.
\newblock {\em Journal of the ACM}, 65(6):35:1--35:26, November 2018.

\bibitem{BraunS24}
Alexander Braun and Sherry Sarkar.
\newblock {The Secretary Problem with Predicted Additive Gap}.
\newblock In Amir Globerson, Lester Mackey, Danielle Belgrave, Angela Fan,
  Ulrich Paquet, Jakub~M. Tomczak, and Cheng Zhang, editors, {\em Proceedings
  of {NeurIPS} 2024}, volume~37, pages 16321--16341, 2024.

\bibitem{Chen0LT25}
Ziyun Chen, Zhiyi Huang, Dongchen Li, and Zhihao~Gavin Tang.
\newblock {Prophet Secretary and Matching: the Significance of the Largest
  Item}.
\newblock In Yossi Azar and Debmalya Panigrahi, editors, {\em Proceedings of
  {SODA} 2025}, pages 1371--1401. {SIAM}, 2025.

\bibitem{CorreaSZ2021}
Jose Correa, Raimundo Saona, and Bruno Ziliotto.
\newblock {Prophet secretary through blind strategies}.
\newblock {\em Mathematical Programming}, 190(1-2):483--521, 2021.

\bibitem{CorreaCFOT21}
Jos{\'{e}}~R. Correa, Andr{\'{e}}s Cristi, Laurent Feuilloley, Tim Oosterwijk,
  and Alexandros Tsigonias{-}Dimitriadis.
\newblock {The Secretary Problem with Independent Sampling}.
\newblock In {\em Proceedings of {SODA} 2021}, pages 2047--2058, 2021.

\bibitem{CorreaCES22}
José Correa, Andrés Cristi, Boris Epstein, and José Soto.
\newblock {The Two-Sided Game of Googol}.
\newblock {\em Journal of Machine Learning Research}, 23(113):1--37, 2022.

\bibitem{DuttingLLV21}
Paul D{\"{u}}tting, Silvio Lattanzi, Renato~Paes Leme, and Sergei
  Vassilvitskii.
\newblock {Secretaries with Advice}.
\newblock {\em Mathematics of Operations Research}, 49(2):856--879, 2024.

\bibitem{Dynkin1963}
E.~B. Dynkin.
\newblock {The optimum choice of the instant for stopping a Markov process}.
\newblock {\em Soviet Math. Dokl}, 4:627--629, 1963.

\bibitem{Ehsani2018}
Soheil Ehsani, MohammadTaghi Hajiaghayi, Thomas Kesselheim, and Sahil Singla.
\newblock {Prophet Secretary for Combinatorial Auctions and Matroids}.
\newblock {\em SIAM Journal on Computing}, 53(6):1641--1662, 2024.

\bibitem{Erlebach2015}
Thomas Erlebach and Michael Hoffmann.
\newblock {Query-Competitive Algorithms for Computing with Uncertainty}.
\newblock {\em Bulletin of EATCS}, 116:22--39, 2015.
\newblock The Algorithmics Column, edited by Gerhard J. Woeginger.

\bibitem{FeldmanSZ18}
Moran Feldman, Ola Svensson, and Rico Zenklusen.
\newblock {A Simple O(log log(rank))-Competitive Algorithm for the Matroid
  Secretary Problem}.
\newblock {\em Mathematics of Operations Research}, 43(2):638--650, May 2018.

\bibitem{FujiiYoshida2024}
Kaito Fujii and Yuichi Yoshida.
\newblock {The Secretary Problem with Predictions}.
\newblock {\em Mathematics of Operations Research}, 49(2):1241--1262, 2024.

\bibitem{GilbertM66}
John~P. Gilbert and Frederick Mosteller.
\newblock {Recognizing the maximum of a sequence}.
\newblock {\em Journal of the American Statistical Association},
  61(313):35--73, 1966.

\bibitem{HoeferK17}
Martin Hoefer and Bojana Kodric.
\newblock {Combinatorial Secretary Problems with Ordinal Information}.
\newblock In {\em Proceedings of {ICALP} 2017}, volume~80 of {\em Leibniz
  International Proceedings in Informatics (LIPIcs)}, pages 133:1--133:14.
  Schloss Dagstuhl--Leibniz-Zentrum fuer Informatik, 2017.

\bibitem{JailletSZ13}
Patrick Jaillet, José~A. Soto, and Rico Zenklusen.
\newblock {Advances on matroid secretary problems: Free order model and laminar
  case}.
\newblock In {\em Proceedings of {IPCO} 2013}, volume 7801 of {\em Lecture
  Notes in Computer Science}, pages 254--265, Valparaiso, Chile, 2013.
  Springer.

\bibitem{Kaplan2022OnlineSample}
Haim Kaplan, David Naori, and Danny Raz.
\newblock {Online Weighted Matching with a Sample}.
\newblock In {\em Proceedings of {SODA} 2022}, pages 1247--1272. {SIAM}, 2022.

\bibitem{KaplanNR20}
Haim Kaplan, David Naori, and Danny Raz.
\newblock {Competitive Analysis with a Sample and the Secretary Problem}.
\newblock {\em SIAM Journal on Computing}, 54(6):1489--1513, 2025.

\bibitem{KarisaniEtAl2026}
Helia Karisani, Mohammadreza Daneshvaramoli, Hedyeh Beyhaghi, Mohammad
  Hajiesmaili, and Cameron Musco.
\newblock {The Secretary Problem with Predictions and a Chosen Order}.
\newblock In Shubhangi Saraf, editor, {\em Proceedings of {ITCS} 2026}, volume
  362 of {\em Leibniz International Proceedings in Informatics (LIPIcs)}, pages
  86:1--86:24, Dagstuhl, Germany, 2026. Schloss Dagstuhl -- Leibniz-Zentrum
  f{\"u}r Informatik.

\bibitem{Lindley1961}
D.~V. Lindley.
\newblock {Dynamic Programming and Decision Theory}.
\newblock {\em Applied Statistics}, 10(1):39--51, March 1961.

\bibitem{LiuMM23Mallows}
Xujun Liu, Olgica Milenkovic, and George~V. Moustakides.
\newblock {Query-based selection of optimal candidates under the Mallows
  model}.
\newblock {\em Theoretical Computer Science}, 979:114206, 2023.

\bibitem{MoustakidesLM24RandomQueries}
George~V. Moustakides, Xujun Liu, and Olgica Milenkovic.
\newblock {Optimal stopping methodology for the secretary problem with random
  queries}.
\newblock {\em Journal of Applied Probability}, 61(2):578--602, 2024.

\bibitem{NourmohammadiCST26}
Hasti Nourmohammadi, Ying Cao, Bo~Sun, and Xiaoqi Tan.
\newblock {Ordinal Secretaries with Advice}.
\newblock In {\em Proceedings of {AAAI} 2026}, pages 37108--37116, 2026.

\bibitem{SotoTV2021}
José~A. Soto, Abner Turkieltaub, and Victor Verdugo.
\newblock {Strong Algorithms for the Ordinal Matroid Secretary Problem}.
\newblock {\em Mathematics of Operations Research}, 46(2):642--673, 2021.

\end{thebibliography}

\end{document}